\documentclass[letter,11pt]{journal}
\usepackage[dvipsnames]{xcolor}
\usepackage{graphicx}
\usepackage{amsthm, amsmath, amssymb, amsfonts}
\usepackage{wrapfig}
\usepackage{enumitem}
\usepackage{xspace}
\usepackage{comment}
\usepackage{tabularx,booktabs}
\usepackage{tcolorbox}
\usepackage{chngcntr}
\usepackage[colorlinks]{hyperref}
\usepackage{microtype}
\usepackage{thmtools}
\usepackage{thm-restate}
\usepackage[left=1in,top=1in,right=1in,bottom=1in]{geometry}
\RequirePackage[left,mathlines]{lineno}

\newcommand*\patchAmsMathEnvironmentForLineno[1]{%
  \expandafter\let\csname old#1\expandafter\endcsname\csname #1\endcsname
  \expandafter\let\csname oldend#1\expandafter\endcsname\csname end#1\endcsname
  \renewenvironment{#1}%
     {\linenomath\csname old#1\endcsname}%
     {\csname oldend#1\endcsname\endlinenomath}}%
\newcommand*\patchBothAmsMathEnvironmentsForLineno[1]{%
  \patchAmsMathEnvironmentForLineno{#1}%
  \patchAmsMathEnvironmentForLineno{#1*}}%
\AtBeginDocument{%
  \patchBothAmsMathEnvironmentsForLineno{equation}%
  \patchBothAmsMathEnvironmentsForLineno{align}%
  \patchBothAmsMathEnvironmentsForLineno{flalign}%
  \patchBothAmsMathEnvironmentsForLineno{alignat}%
  \patchBothAmsMathEnvironmentsForLineno{gather}%
  \patchBothAmsMathEnvironmentsForLineno{multline}}

\graphicspath{{./figures/}}
\hypersetup{allcolors = black}

\newtheorem{theorem} {Theorem}
\newtheorem{lemma}[theorem] {Lemma}
\newtheorem{corollary}[theorem] {Corollary}
\newtheorem{observation}[theorem] {Observation}

\renewcommand{\subparagraph}[1]{\smallskip\noindent\textbf{\sffamily #1}}

\newcommand{\frechet}{Fr\'echet\xspace}
\newcommand{\etal}{\textnormal{et al.}\xspace}

\newcommand{\mkmcal}[1]{\ensuremath{\mathcal{#1}}\xspace}

\newcommand{\T}{\mkmcal{T}}

\newcommand{\D}{\mkmcal{D}}

\newcommand{\F}{\mkmcal{F}}

\newcommand{\B}{\mkmcal{B}}
\newcommand{\DB}{\mkmcal{D}_B}
\newcommand{\xdistance}{\delta}

\newcommand{\fd}{\ensuremath{\D_\F}\xspace}

\newcommand{\dhd}{\ensuremath{\overrightarrow{\D}_{\!H}\xspace}}

\newcommand{\from}{\colon}
\newcommand{\Skew}{\ensuremath{P^{\le}}}
\newcommand{\CH}{\ensuremath{\mathit{CH}}\xspace}

\newcommand{\FSVD}{\mathit{FSVD}\xspace}

\newcommand{\mkmbb}[1]{\ensuremath{\mathbb{#1}}\xspace}

\newcommand{\LL}[1]{\ensuremath{\protect\overleftarrow{#1}}\xspace}
\newcommand{\RR}[1]{\ensuremath{\protect\overrightarrow{#1}}\xspace}

\newcommand{\LLa}[1]{\ensuremath{\overset{\leftarrow\,\alpha\vspace{-1ex}}{#1}}\xspace}
\newcommand{\RRa}[1]{\ensuremath{\overset{\rightarrow\,\alpha}{#1}}\xspace}

\newcommand{\R}{\mkmbb{R}}

\newcommand{\eps}{\ensuremath{\varepsilon}\xspace}

\newcommand{\thmheadfont}{\bfseries}
\newenvironment{repeatenv}[2]%
  {\smallskip\noindent {\thmheadfont #1~\ref{#2}.}\ \slshape}
  {\normalfont}

\title{Efficient \frechet distance queries for segments}

\author{Maike Buchin\thanks{
    Ruhr University Bochum, Bochum, Germany, \texttt{maike.buchin@rub.de}
  }
  \and Ivor van der Hoog\thanks{
    Utrecht University, Utrecht, The Netherlands \texttt{[i.d.vanderhoog,t.a.e.ophelders,f.staals]@uu.nl}
  }
  \and Tim Ophelders\footnotemark[2]~\thanks{
    TU Eindhoven, Eindhoven, The Netherlands
  }
  \and Lena Schlipf\thanks{
    Universit\"at T\"ubingen, T\"ubingen, Germany, \texttt{schlipf@informatik.uni-tuebingen.de}
  }
  \and Rodrigo I. Silveira\thanks{
    Universitat Polit\`{e}cnica de Catalunya, Barcelona, Spain \texttt{rodrigo.silveira@upc.edu}
  }
  \and Frank Staals\footnotemark[2]
}

\begin{document}

\maketitle

\begin{abstract}
We study the problem of constructing a data structure that can store a two-dimensional polygonal curve $P$, such that for any query segment $\overline{ab}$ one can efficiently compute the Fr\'{e}chet distance between $P$ and $\overline{ab}$. First we present a data structure of size $O(n \log n)$ that can compute the Fr\'{e}chet distance between $P$ and a horizontal query segment $\overline{ab}$ in $O(\log n)$ time, where $n$ is the number of vertices of $P$. In comparison to prior work, this significantly reduces the required space. We extend the type of queries allowed, as we allow a query to be a horizontal segment $\overline{ab}$ together with two points $s, t \in P$ (not necessarily vertices), and ask for the Fr\'{e}chet distance between $\overline{ab}$ and the curve of $P$ in between $s$ and $t$. Using $O(n\log^2n)$ storage, such queries take $O(\log^3 n)$ time, simplifying and significantly improving previous results. We then generalize our results to query segments of arbitrary orientation. We present an $O(nk^{3+\varepsilon}+n^2)$ size data structure, where $k \in [1..n]$ is a parameter the user can choose, and $\varepsilon > 0$ is an arbitrarily small constant, such that given any segment $\overline{ab}$ and two points $s, t \in P$ we can compute the Fr\'{e}chet distance between $\overline{ab}$ and the curve of $P$ in between $s$ and $t$ in $O((n/k)\log^2n+\log^4 n)$ time. This is the first result that allows efficient exact Fr\'{e}chet distance queries for arbitrarily oriented segments.

We also present two applications of our data structure: we show that we can compute a local $\delta$-simplification (with respect to the Fr\'{e}chet distance) of a polygonal curve in $O(n^{5/2+\varepsilon})$ time, and that we can efficiently find a translation of an arbitrary query segment $\overline{ab}$ that minimizes the Fr\'{e}chet distance with respect to a subcurve of $P$.
\end{abstract}
\clearpage
\setcounter{page}{1}

\section{Introduction}
\label{sec:Introduction}


Comparing the shape of polygonal curves is an important task that
arises in many contexts such as GIS
applications~\cite{alt2003matching,buchin2013median}, protein
classification~\cite{jiang2008protein}, curve
simplification~\cite{buchin2011detecting}, curve
clustering~\cite{agarwal2005near} and even speech
recognition~\cite{kwong1998parallel}. Within computational geometry,
there are two  well studied distance measures for polygonal curves: the
Hausdorff and the Fr\'{e}chet distance.  
The \frechet distance has proven particularly 
useful 
as it takes the course of the curves into account. 
However, the \frechet distance between curves is costly to compute, as its computation requires roughly quadratic time~\cite{ag-cfdbt-95,buchin17four_soviet_walk_dog}. 
When a large number of \frechet distance queries are required, we
would like to have a data structure, a so-called \emph{distance oracle}, to answer these queries more
efficiently.

This leads to a fundamental data structuring problem: preprocess a polygonal curve such that, given a query polygonal curve, their Fréchet distance can be computed efficiently.
(Here, the query curves are assumed to be of small size compared to the preprocessed one.)
It turns out that this problem is extremely challenging. 
Firstly, polygonal curves are complex objects of non-constant complexity, where most operations have a non-negligible cost.
Secondly, the high computational cost of the Fréchet distance complicates all attempts to design exact algorithms for the problem.
Indeed, even though great efforts have been devoted to design data structures to answer Fréchet distance queries, there is still no distance oracle known that is able to answer \emph{exactly} queries for a general query curve.

To make progress on this important problem, in this work we focus on a more restrictive but fundamental setting: when the query curve is a single segment. 
The reasons to study this variant of the problem are twofold.
On the one hand, it is a necessary step to solve the general problem.
On the other hand, it is a setting that has its own applications.  For example, in trajectory simplification, or when trying to find subtrajectories that are geometrically close to a given query segment (e.g. when computing shortcut-variants of the \frechet distance~\cite{driemel13:jaywalk}, or in    
trajectory analysis~\cite{de2013fast} on sports data). A similar strategy of tackling segment queries has also been successfully applied in nearest neighbor queries with the \frechet distance~\cite{aronov2019:nearest_neigbhbor_queries_planar_curves}.


More precisely, we study the problem of preprocessing a polygonal curve $P$ 
to determine the \emph{exact} \emph{continuous} \frechet distance between $P$ and a query segment in
sublinear time. Specifically, we study preprocessing a polygonal curve $P$ of $n$ vertices in the plane, such that given a query segment $\overline{ab}$, traversed from $a$ to~$b$, the Fr\'{e}chet distance between $P$ and $\overline{ab}$ can be computed in sublinear time. 
Note that without preprocessing, this problem can be solved in  $O(n\log n)$ time. 

\subparagraph{Related work.} Data structures that support
(approximate) nearest neighbor queries with respect to the \frechet
distance have received considerable attention throughout the
years, see for instance, these recent papers~\cite{aronov2019:nearest_neigbhbor_queries_planar_curves,driemelpsarros2020,filtser2020approximateICALP} and the references therein. 
In these problems, the goal is typically to store a set of polygonal
curves such that given a query curve and a query threshold $\Delta$
one can quickly report (or count) the curves that are within (discrete) \frechet distance $(1+\eps)\Delta$, for some $\eps > 0$, of the query curve. Some of these data structures even allow
approximately counting the number of curves that have a subcurve
within \frechet distance $\Delta$~\cite{de2013fast}. Also highlighting its practical importance, the near neighbor problem using \frechet distance was posed as ACM Sigspatial GIS Cup in 2017~\cite{gis-cup-18}.


Here we consider the problem in which we want to compute the \frechet
distance of (part of) a curve to a low complexity \emph{query} curve. For the
\emph{discrete} \frechet distance, efficient $(1+\eps)$-approximate
distance oracles are known, even when $P$ is given in an online
fashion~\cite{filtser2021static}. For the \emph{continuous} \frechet
distance that we consider the results are more restrictive.
Driemel and Har-Peled~\cite{driemel13:jaywalk} present an
$O(n \eps^{-4} \log \eps^{-1})$ size data structure that given a query
segment $\overline{ab}$ can compute a $(1+\eps)$-approximation of the
\frechet distance between $P$ and $\overline{ab}$ in
$O(\eps^{-2} \log n \log\log n)$ time. Their approach extends to
higher dimensions and low complexity polygonal query curves.

Gudmundsson~\etal~\cite{Gudmundsson:longJournal} present an
$O(n\log n)$ sized data structure that can \emph{decide} if the
\frechet distance to $\overline{ab}$ is smaller than a given value
$\Delta$ in $O(\log^2 n)$ time (so with some parametric search
approach one could consider computing the \frechet distance
itself). However, their result holds only when the length of
$\overline{ab}$ and all edges in $P$ is relatively large compared to
$\Delta$. De Berg~\etal~\cite{de2017data} presented an $O(n^2)$ size
data structure that does not have any restrictions on the length of
the query segment or the edges of $P$. However, the orientation of the
query segment is restricted to be horizontal. Queries are supported in
$O(\log^2 n)$ time, and even queries asking for the \frechet distance
to a vertex-to-vertex subcurve are allowed (in that case, using
$O(n^2\log^2 n)$ space). Recently,
Gudmundsson~\etal~\cite{gudmundssonfrechet2020} extended this result
to allow the subcurve to start and end anywhere within $P$. Their data
structure has size $O(n^2 \log^2 n)$ and queries take $O(\log^8 n)$
time. In their journal version,
Gudmundsson~\etal~\cite{gudmundssonFrechetJournal} directly apply the
main result of a preliminary version of this
paper~\cite{eurocgversion} to immediately improve space usage of their
data structure to $O(n^{3/2})$; their preprocessing time remains
$O(n^2\log^2n)$. The current version of this paper significantly
improves these results. Moreover, we present data structures that
allow for arbitrarily oriented query segments.

\subparagraph{Problem statement \& our results.} Let $P$ be a
polygonal curve in $\R^2$ with $n$ vertices $p_1,\dots,p_n$. For ease of
exposition, we assume that the vertices of $P$ are in general
position, i.e., all $x$- and $y$-coordinates are unique, no three
points lie on a line, and no four points are cocircular. We consider
$P$ as a function mapping any time $t \in [0,1]$ to a point $P(t)$ in
the plane. Our ultimate goal is to store $P$ such that we can quickly compute
the \emph{\frechet distance} $\fd(P,Q)$ between $P$ and a query curve
$Q$. The \frechet distance is defined as
\[ \fd(P,Q) = \inf_{\alpha,\beta}\max_{t \in [0,1]} \|P(\alpha(t))-Q(\beta(t))\|, \]
\noindent
where $\alpha,\beta\from[0,1]\to[0,1]$ are nondecreasing surjections,
also called reparameterizations of $P$ and $Q$, respectively, and
$\|p-q\|$ denotes the Euclidean distance between $p$ and $q$.

In this work we focus on the case where $Q$ is a single line segment $\overline{ab}$
starting at $a$ and ending at $b$.
Note that $P$ may self-intersect and $\overline{ab}$ may intersect $P$. 
Our first main result deals with the case where 
$\overline{ab}$ is horizontal:

\begin{restatable}{theorem}{horizontalfull}
  \label{thm:horizontal_full_curve_ds}
  Let $P$ be a polygonal curve in $\R^2$ with $n$ vertices. There is
  an $O(n\log n)$ size data structure that can be built in
  $O(n\log^2 n)$ time such that given a horizontal query segment
  $\overline{ab}$ it can report $\fd(P,\overline{ab})$ in $O(\log n)$
  time.
\end{restatable}

This significantly improves over the earlier result of de Berg
\etal~\cite{de2017data}, as we reduce the required space and
preprocessing time from quadratic to near linear. We simultaneously
improve the query time from $O(\log^2 n)$ to $O(\log n)$. We further
extend our results to allow queries against subcurves of $P$. Let
$s, t$ be two points on $P$, we use $P[s,t]$ to denote the subcurve of $P$ from
$s$ to $t$. For horizontal query segments we then get:

\begin{theorem}
  \label{thm:horizontal_subcurve_ds}
  Let $P$ be a polygonal curve in $\R^2$ with $n$ vertices. There is
  an $O(n\log^2 n)$ size data structure that can be built in
  $O(n\log^2 n)$ time such that given a horizontal query segment
  $\overline{ab}$ and two query points $s$ and $t$ on $P$ it can
  report $\fd(P[s,t],\overline{ab})$ in $O(\log^3 n)$ time.
\end{theorem}

De Berg \etal presented a data structure that could
handle such queries in $O(\log^2 n)$ time (using $O(n^2\log^2 n)$
space), provided that $s$ and $t$ were vertices of $P$. Compared to
their data structure we thus again significantly improve the space
usage, while allowing more general queries. 
The recently presented data structure of Gudmundsson
\etal~\cite{gudmundssonfrechet2020} does allow $s$ and $t$ to lie on
the interior of edges of $P$ (and thus supports queries against
arbitrary subcurves). Their original data structure uses $O(n^2\log^2 n)$ space
and allows for $O(\log^8 n)$ time queries. Compared to their
result we use significantly less space, while also
improving the query time.

Using the insights gained in this restricted setting, we then present
the first data structure that allows exact \frechet distance queries
with arbitrarily oriented query segments in sublinear time. With near
quadratic space we obtain a query time of $O(n^{2/3}\log^2 n)$. If we
insist on logarithmic query time the space usage increases to
$O(n^{4+\eps})$. In particular, we present a data structure with the
following time-space tradeoff. 
At only a small additional cost we can also support subcurve queries. 

\begin{theorem} 
  \label{thm:arbitrary_full_curve_ds}
  Let $P$ be a polygonal curve in $\R^2$ with $n$ vertices, and let
  $k \in [1..n]$ be a parameter. There is an $O(nk^{3+\eps}+n^2)$ size
  data structure that can be built in $O(nk^{3+\eps}+n^2)$ time such that
  given an arbitrary query segment $\overline{ab}$ it can report
  $\fd(P,\overline{ab})$ in $O((n/k)\log^2 n)$ time. 
  In addition, given two query points $s$ and $t$ on $P$, it can
  report $\fd(P[s,t],\overline{ab})$ in $O((n/k)\log^2 n + \log^4 n)$ time. 
\end{theorem}

In Theorem~\ref{thm:arbitrary_full_curve_ds} and throughout the rest of the paper $\eps > 0$ is an arbitrarily small constant. 
%
In both Theorem~\ref{thm:horizontal_subcurve_ds} and
Theorem~\ref{thm:arbitrary_full_curve_ds} the query time can be made
sensitive to the number of vertices $m=|P[s,t]|$ in the query subcurve
$P[s,t]$. That is, we can get query times $O(\log^3 m)$ and
$O((m/k)\log^2 m + \log^4 m)$, respectively.

To achieve our results, we also develop data structures that allow us to
efficiently query the \emph{directed Hausdorff distance}
$\dhd(P[s,t],\overline{ab}) = \max_{p \in P[s,t]} \min_{q \in
  \overline{ab}} \|p-q\|$ from (a subcurve $P[s,t]$ of) $P$ to the
query segment $\overline{ab}$.
For an arbitrarily oriented query
segment $\overline{ab}$ and a query subcurve $P[s,t]$ our data
structure uses $O(n\log n)$ space and can answer such queries in
$O(\log^2 n)$ time. 
Using more space, queries can be answered in $O(\log n)$ time, see Section~\ref{sec:Arbitrary_Query_Orientation}.

\subparagraph{Applications.}
In Section~\ref{sec:Applications} we show how to efficiently solve two problems using our data structure. First, we show
how to compute a local $\delta$-simplification of $P$---that is, a
minimum complexity curve whose edges are within \frechet distance
$\delta$ to the corresponding subcurve of $P$---in $O(n^{5/2+\eps})$
time. This improves existing $O(n^3)$ time
algorithms~\cite{godau91}. Second, given a query segment
$\overline{ab}$ we show how to efficiently find a translation of
$\overline{ab}$ that minimizes the \frechet distance to (a given
subcurve of) $P$. This extends the work of
Gudmundsson~\etal~\cite{gudmundssonFrechetJournal} to arbitrarily
oriented segments. Furthermore, we answer one of their open problems by showing how to find a scaling
of the segment that minimizes the \frechet distance.

\section{Global approach}
\label{sec:Global_Approach}

We illustrate the main ideas of our approach, in particular for the
case where the query segment $\overline{ab}$ is horizontal, with $a$
left of $b$. We can build a symmetric data structure in case $a$ lies
right of $b$. We now first review some definitions based on
those in~\cite{de2017data}.

\begin{figure}[t]
  \centering
  \includegraphics{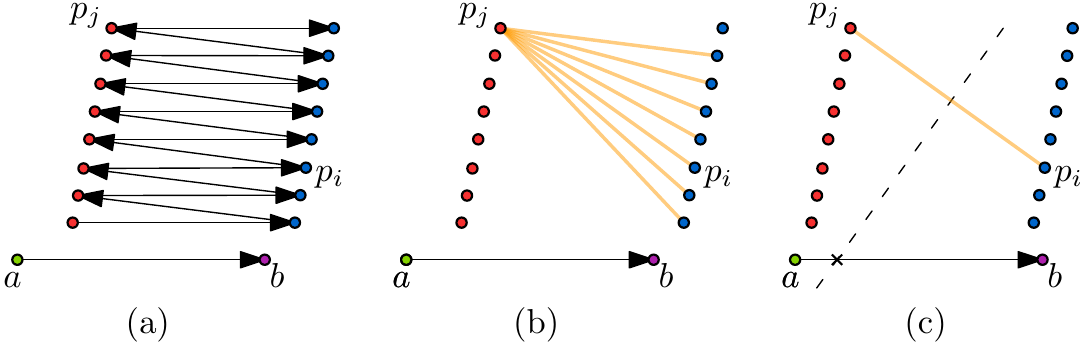}
  \caption{
    (a) A polygonal curve and query segment.
    (b) The red vertex $p_j$ forms a backward pair with all but one blue vertex.
    (c) For a fixed backward pair $(p_i, p_j)$,
    we consider 
    the distance between the  intersection (cross) of their bisector (dashed) and $\overline{ab}$, and either $p_i$ or~$p_j$.
  }
  \label{fig:backwardpair}
\end{figure}

Let $\Skew$ be the set of ordered pairs of vertices
$(p, q) \in P \times P$ where $p$ precedes or equals $q$ along $P$. An
ordered pair $(p,q) \in \Skew$ forms a \emph{backward pair} if
$x_{q}\leq x_{p}$. Here, and throughout the rest of the paper, $x_p$ and
$y_p$ denote the $x$- and $y$-coordinates of point $p$,
respectively. The set of all backward pairs of $P$ will be denoted
$\mathcal{B}_{}(P)$. A backward pair $(p,q)$ is \emph{trivial} if
$p=q$.  
See Fig.~\ref{fig:backwardpair} for an example of backward pairs
(omitting trivial pairs). For two points $p, q \in P$, we then define
$ \xdistance_{pq}(y) = \min_x \max \left\{\|(x,y)-p\|,\|(x,y)-q\|
\right\}.  $ That is, $\xdistance_{pq}(y)$ is a function that for any
$y$-coordinate gives the minimum possible distance between a point at
height $y$ and both $p$ and~$q$.  We will use the function
$\xdistance_{pq}$ only when $(p,q) \in \B(P)$ is a backward pair. We
then define the function
$ \DB(y) = \max\left\{ \xdistance_{pq}(y) \mid
  (p,q)\in\mathcal{B}(P)\right\}, $ which we refer to as the
\emph{backward pair distance} of a horizontal segment at height $y$
with respect to $P$. Note that $\DB(y)$ is the upper envelope of the
functions $\xdistance_{pq}$ for all backward pairs $(p,q)$ of $P$.

De Berg \etal~\cite{de2017data} prove that the \frechet
distance is the maximum of four terms:
\begin{equation}
  \fd(P, \overline{ab}) = \max \left\{ || p_1-a ||, \quad || p_n-b ||, \quad \dhd(P, \overline{ab}), \quad \DB(y_a) \right\}.\label{eq:decompose_FD}
\end{equation}

The first two terms are trivial to compute in $O(1)$
time. Like de Berg \etal, we build separate data structures that allow
us to efficiently compute the third and fourth terms.

A key insight is that we can compute $\dhd(P,\overline{ab})$ by
building the furthest segment Voronoi diagrams (FSVD) of two sets of
horizontal halflines, and querying these diagrams with the endpoints
$a$ and $b$. See
Section~\ref{sub:Horizontal:_The_Hausdorff_Term}. This allows for a
linear space data structure that supports querying
$\dhd(P,\overline{ab})$ in $O(\log n)$ time, improving both the space
and query time over~\cite{de2017data}.

However, in~\cite{de2017data} the data structure that supports
computing the backward pair distance dominates the required space
and preprocessing time, as there may be $\Omega(n^2)$ backward pairs,
see Fig.~\ref{fig:backwardpair}. Via a divide and conquer argument
we show that the number of backward pairs that show up on the upper
envelope $\DB$ is only $O(n \log n)$, see
Section~\ref{sub:Horizontal:_Backward_Pairs}. The crucial ingredient
is that there are only $O(n)$ backward pairs $(p,q)$ contributing to
$\DB$ in which $p$ is a vertex among the first $n/2$ vertices 
of $P$, and $q$ is a vertex in the remaining $n/2$
vertices. Surprisingly, we can again argue this using furthest segment Voronoi
diagrams of sets of horizontal halflines. This allows us
to build $\DB$ in $O(n\log^2 n)$ time in total. In 
Section~\ref{sec:Horizontal:_Subtrajectories} we 
extend these results to support queries against an arbitrary subcurve $P[s,t]$ of~$P$.

For arbitrarily oriented query segments we similarly decompose
$\fd(P,\overline{ab})$ into four terms, and build a data structure for
each term separately, see Section~\ref{sec:Arbitrary_Query_Orientation}. The directed Hausdorff term can still be queried efficiently using an $O(n\log^2 n)$ size data structure. However, our initial data structure for the backward pair distance uses $O(n^{4+\eps})$ space. The main reason for this
is that functions $\delta_{pq}$ expressing the cost of a backward pair
are now bivariate, depending on both the slope and intercept of the
supporting line of $\overline{ab}$. The upper envelope of a set of $n$
such functions may have quadratic complexity. While our divide and
conquer strategy does not help us to directly bound the complexity of the
(appropriately generalized function) $\DB$ in this case, it does allow
us to support queries against subcurves of $P$. Moreover, we can use it to obtain a favourable query time vs. space trade off. In Section~\ref{sec:Applications} we then apply our data structure to efficiently solve various \frechet distance related problems. 

\section{Horizontal queries}
\label{sec:Horizontal}

\subsection{The Hausdorff term}\label{sub:Horizontal:_The_Hausdorff_Term}
	\begin{figure}[tb]
	\centering
	\includegraphics{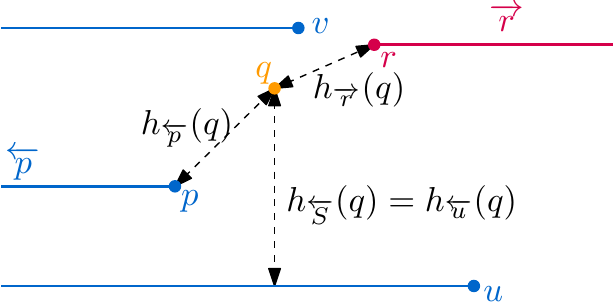}
	\caption{Example for $S$=$\{p,u,v\}$.  }
	\label{fig:functions_hs}
\end{figure}

In this section we show that there is a linear-size data structure to query the
Hausdorff term in $O(\log n)$ time, which can be built in $O(n\log n)$ time.

For a point $p \in \R^2$, define $\LL{p}$ to be the ``leftward''
horizontal halfline starting at $p$ and containing all points
directly to the left of $p$.
Analogously, we define $\RR{p}$ as the
``rightward'' horizontal halfline starting at $p$, so that
$p=\LL{p} \cap \RR{p}$.  We extend this notation to any set of points $S$,
that is, $\LL{S}=\{\LL{s} \mid s \in S\}$ denotes the set of
``leftward'' halflines starting at the points in $S \subseteq
\R^2$.
We define $\RR{S}$ analogously.
Let $S$ and $T$ be two (possibly overlapping) point sets in the plane.
We define the following distance functions (see also Fig.~\ref{fig:functions_hs}.):
\begin{align*}
h_{\LL{p}}(q) &= \dhd(\{q\},\LL{p}) = \min\left\{\|p'-q\|\mid p' \in \LL{p}\right\}, 
\qquad h_{\LL{S}}(q) = \max\left\{h_{\LL{p}}(q) \mid p\in S\right\}\\
h_{\RR{p}}(q) &= \dhd(\{q\},\RR{p}) = \min\left\{\|p'-q\|\mid p' \in \RR{p}\right\},
\qquad h_{\RR{S}}(q) = \max\left\{h_{\RR{p}}(q) \mid p\in S\right\}.
\end{align*}

Note that $h_{\RR{S}}$ (resp., $h_{\LL{S}}$) is the upper envelope of the distance
functions to the halflines in $\RR{S}$ (resp., $\LL{S}$).
Since $h_{\RR{S}}$ and $h_{\LL{S}}$ map each point in the plane to a distance, the envelopes live in $\R^3$.

We begin by providing some observations on the computation of the Hausdorff term.

\begin{observation}
	\label{obs:fsvd}
	Let $S$ be a set of points in the plane. The (graph of the) function $h_{\RR{S}}$ is an upper envelope whose orthogonal projection is the furthest segment Voronoi diagram of the set of halflines $\RR{S}$. A symmetrical property holds for $h_{\LL{S}}$. 
\end{observation}

For points $p=(x_p,y_p)$ and $q=(x_q,y_q)\in\mathbb{R}^2$, we have
\[
h_{\LL{p}}(q)=
\begin{cases}
\|p-q\| &\text{if $x_p\leq x_q$}\\|y_p-y_q|&\text{if $x_p\geq
	x_q$,}
\end{cases} \text{\qquad and\qquad}
h_{\RR{p}}(q)=
\begin{cases}\|p-q\| &\text{if $x_p\geq x_q$}\\|y_p-y_q|&\text{if
	$x_p\leq x_q$.}
\end{cases}
\]

\begin{observation}
	\label{obs:monotonic}
	For any fixed $y$ and $p \in S$, the function $x\mapsto h_{\LL{p}}(x,y)$ for a point $p$ is monotonically increasing, and $x\mapsto h_{\RR{p}}(x,y)$ is
	monotonically decreasing.  Consequently, also for any point set $S$, the function $x\mapsto h_{\RR{S}}(x,y)$
	is monotonically decreasing, and $x\mapsto h_{\LL{S}}(x,y)$ is
	monotonically increasing.
\end{observation}

\begin{lemma}\label{lem:dHtoSegment}
	For any horizontal segment $\overline{ab}$ and
	any point set $S\subseteq\mathbb{R}^2$, we have
	$\dhd(S,\overline{ab})=\max\left\{h_{\LL{S}}(a),h_{\RR{S}}(b)\right\}$.
\end{lemma}

\begin{proof}
	Assume w.l.o.g. that $x_a\leq x_b$, and let $y = y_a = y_b$.
	Partition $S$ into three disjoint subsets $L$, $R$, and $M$, where
	$L \subseteq S$ contains all points in $S$ strictly left of $a$, $R$
	all points strictly right of $b$, and $M$ all (remaining) points that lie in
	the vertical slab defined by $x_a$ and $x_b$. We then have:
	\begin{align*}
	\dhd(S,\overline{ab})
	= \max_{s\in S}\min_{r\in\overline{ab}}\|s-r\|
	&= \max\left\{\sup_{s\in L} \|s-a\|
	, \max_{s\in R} \|s-b\|
	, \max_{s \in M} |y_s-y|
	\right\} \\
	&= \max\left\{h_{\LL{L}}(a), h_{\RR{R}}(b), \max_{s \in M} |y_s-y|\right\},
	\end{align*}

	\noindent
	For all points $s$ not strictly left of $a$, so
	in particular those in $M$, we have that
	$h_{\LL{s}}(a) = |y_s-y_a| = |y_s-y|$, and thus
	$\sup_{s \in M} |y_s-y| = h_{\LL{M}}(a)$. For the points $s \in R$
	(these are the points right of $a$, that even lie right of $b$) we have
	$h_{\LL{s}}(a) = |y_s-y_a| = |y_s- y_b| \leq \|s-b\| = h_{\RR{s}}(b)$. It
	therefore follows that $h_{\LL{R}}(a) \leq h_{\RR{R}}(b)$.
	Consequently  we observe that by definition: $\dhd(S,\overline{ab})
	= \max\left\{h_{\LL{L}}(a), h_{\RR{R}}(b), \sup_{s \in M} |y_s-y|\right\}$. 
	Finally we conclude:
	
	\[
	\dhd(S,\overline{ab}) = \max\left\{h_{\LL{S}}(a), h_{\RR{R}}(b)\right\}.
	\]

	Symmetrically, we obtain that
	$\sup_{s \in M} |y_s-y| = h_{\RR{M}}(b)$, and that
	$h_{\RR{L}}(b) \leq h_{\LL{L}}(a)$. Therefore
	$\dhd(S,\overline{ab}) = \max\left\{h_{\LL{L}}(a), h_{\RR{S}}(b)\right\}$. The lemma
	now follows since $L,R \subseteq S$.
\end{proof}

\begin{corollary}\label{lem:dHtoPoint}
	For any point $p$ and point set $S\subseteq\mathbb{R}^2$,
	$\dhd(S,p)=\max\left\{h_{\LL{S}}(p),h_{\RR{S}}(p)\right\}$.
\end{corollary}

By Observation~\ref{obs:fsvd}, $h_{\LL{S}}$ corresponds to the
furthest segment Voronoi diagram (FSVD) of $\overleftarrow{S}$.  For
$d=2$, this diagram has size $O(n)$ and can be computed in
$O(n\log n)$ time~\cite{PapadopoulouD13}. Thus, by preprocessing the
FSVD for planar point location
queries~\cite{sarnak86planar_point_locat_using_persis_searc_trees} we
obtain a linear space data structure that allows us to evaluate
$h_{\LL{S}}(q)$ for any query point $q \in \R^2$ in $O(\log n)$
time. Note that Sarnak and
Tarjan~\cite{sarnak86planar_point_locat_using_persis_searc_trees}
define their point location data structure (a sweep with a partially
persistent red-black tree) for planar subdivisions whose edges are
line segments. However, their result directly applies to subdivisions
with $x$-monotone curved segments of low algebraic degree. Since the
edges in the FSVD are line segments, rays, or parabolic
arcs~\cite{PapadopoulouD13} we can easily split each such edge into
$O(1)$ $x$-monotone curved segments, and thus use their result as
well. Analogously, we build a linear space data structure for querying
$h_{\RR{S}}$, and obtain the following result through
Lemma~\ref{lem:dHtoSegment}.

\begin{theorem}
  \label{thm:dHtoPoint_algo}
  Let $S$ be a set of $n$ points in $\R^2$. In $O(n\log n)$ time we
  can build a data structure of linear size so that given a
  horizontal query segment $\overline{ab}$,
  $\dhd(S,\overline{ab})$ can be computed in $O(\log n)$ time.
\end{theorem}
Note that the directed Hausdorff distance from a polygonal curve $P$ to a (horizontal) line segment is attained at a vertex of $P$~\cite{de2017data}, thus, we can use Theorem~\ref{thm:dHtoPoint_algo} to compute it.

\subsection{The backward pairs term}
\label{sub:Horizontal:_Backward_Pairs}

In this section we show that the function $\DB$, representing the
backward pair distance, has complexity $O(n\log n)$, can be
computed in $O(n\log^2 n)$ time, and can be evaluated for
some query value $y$ in $O(\log n)$ time. 
This leads to an
efficient data structure for querying $P$ for the \frechet distance to
a horizontal query segment $\overline{ab}$, proving
Theorem~\ref{thm:horizontal_full_curve_ds}. 

Recall that $\DB(y)$ is the maximum over all function values  $\delta_{pq}(y)$ for all backward pairs $(p, q) \in \mathcal{B}(P)$. 
To avoid computing $\mathcal{B}(P)$, we define a new function $\delta'_{pq}(y)$ that applies to any ordered pair of points $(p, q) \in \Skew$. We show that for all backward pairs $(p, q) \in \mathcal{B}(P)$, we have $\delta'_{pq}(y) = \delta_{pq}(y)$. For any pair $(p, q) \in \Skew$ that is not a backward pair, we show that there exists a backward pair $(p', q') \in \mathcal{B}(P)$ such that $\delta'_{pq}(y) \le \delta'_{p'q'}(y) = \delta_{p'q'}(y)$. Consequently, we can compute the value $\DB(y)$ by computing the maximum value of $\delta'_{pq}(y)$ over all pairs in $\Skew$.
We will show how to do this in an efficient manner.

\subsubsection{Decomposing the backward pair distance}

Recall that for any two points $(p, q)$, we denote by $\LL{q}$ the leftward horizontal ray originating from $q$ and by $\RR{p}$ the rightward horizontal ray originating from $p$. 
For each pair of points $(p, q) \in \Skew$, we define the \emph{pair distance} between a query $\overline{ab}$ at height $y$ and $(p,q)$ as the Hausdorff distance from a horizontal line of height $y$ to $(\RR{p} \cup \LL{q})$, or more formally:
\begin{equation*}
    \delta'_{pq}(y) = \min_x \max \left\{h_{\LL{q}}( (x, y) ),
      h_{\RR{p}}( (x, y ) ) \right\}.
\end{equation*}

 First we analyze the case where $(p,q)$ is a backward pair.   

 \begin{lemma}
 	\label{lem:backward}
 	Let $(p, q) \in \Skew$ be a pair of points with $x_p \ge x_q$. Then for all $y$, $\delta_{pq} (y) = \delta'_{pq} (y)$.
 \end{lemma}

 \begin{proof}
 	For all points $p$, the function $x \mapsto\|(x,y)-p\|$ is convex, and minimized at $x=x_p$.
 	The function $\delta_{pq}(y)$ is defined as  $\delta_{pq}(y) = \min_x \max \left\{\|(x,y)-p\|,\|(x,y)-q\| \right\}$.
 	Since $x_q \leq x_p$, it follows from the convexity that the function value $\max \left\{\|(x,y)-p\|,\|(x,y)-q\| \right\}$ is minimized for an $x$-coordinate in $[x_q, x_p]$ (and thus $\delta_{pq}(y)$ is realized by an $x$-value in $[x_q, x_p]$).
 	For all points $p$, $\|(x,y)-p\| = \max \left\{h_{\RR{p}}(( x, y ) ), h_{\LL{p}}( (x, y) )\right \}$, thus we observe that:
 	\begin{align*}
 	\xdistance_{pq}(y)
 	&= \min_{x\in[x_q,x_p]} \max \left\{\|(x,y)-p\|,\|(x,y)-q\| \right\}\\
 	&= \min_{x\in[x_q,x_p]} \max \left\{h_{\LL{q}}( (x,y ) ),h_{\LL{q}}( (x,y) ),h_{\LL{p}}( (x,y ) ),h_{\RR{p}}( (x,y) ) \right\}\\
 	&= \min_{x\in[x_q,x_p]} \max \left\{h_{\LL{q}}( (x,y) ),h_{\RR{p}}( (x,y)) \right\} \textnormal{ (see Fig.~\ref{fig:distance})}\\
 	&= \min_x \max \left\{h_{\LL{q}}( (x,y) ),h_{\RR{p}}( (x,y)) \right\}\\
 	&= \xdistance_{pq}'(y).
 	\end{align*}
 \end{proof}

 \begin{figure}[tb]
 	\centering
 	\includegraphics{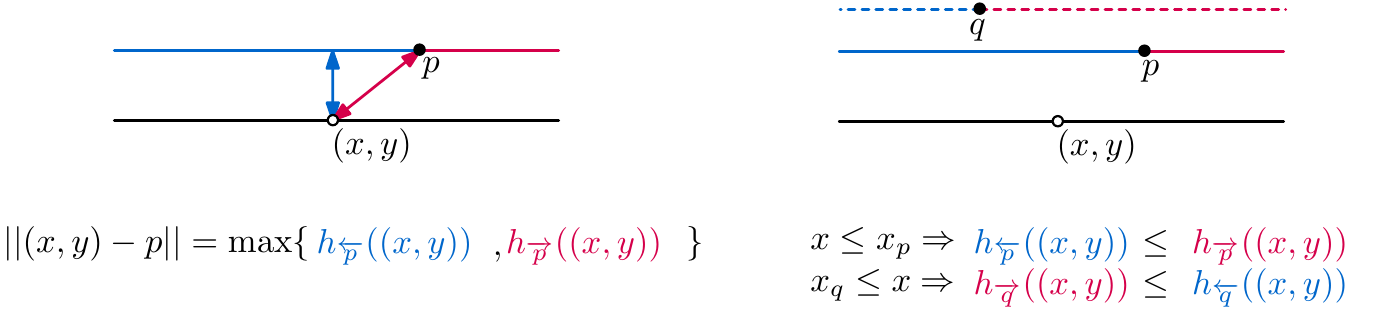}
 	\caption{An illustration of the argument of Lemma~\ref{lem:backward}. 
 	}
 	\label{fig:distance}
 \end{figure}
 
 \noindent
 The consequence of the above lemma is that for each $(p,q) \in \mathcal{B}(P)$, $\xdistance_{pq}(y) = \xdistance_{pq}'(y)$. 
 
 Next we make an observation about pairs of points that are not a backward pair: 
 
 \begin{lemma}\label{lem:nonbackward}
 	Let $(p, q) \in \Skew$ be a pair of points with $x_p < x_q$, then  $\xdistance_{pq}'(y)=\max \left\{\xdistance_{qq}(y),\xdistance_{pp}(y) \right\}$. 
 \end{lemma}
 
 \begin{proof}
 	By definition, $\xdistance_{pq}'(y) = \min_x \max \left\{h_{\LL{q}}( (x, y) ), h_{\RR{p}}( (x, y ) ) \right\}$. For any $y$, the function $x \mapsto h_{\LL{q}}( (x, y) )$ is minimal and constant for all $x \le x_q$.
 	Similarly, the function $x \mapsto h_{\RR{p}}( (x, y) )$ is minimal and constant for all $x_p \le x$. Since $(-\infty, x_q] \cap [x_p, \infty) = [x_p, x_q]$, it follows that:
 	\begin{align*}
 	\xdistance_{pq}'(y) &= \min_x \max \left\{h_{\LL{q}}( (x, y) ), h_{\RR{p}}( (x, y ) ) \right\} \\
 	&= \min_{x \in [x_p, x_q]} \max \left\{h_{\LL{q}}( (x, y) ), h_{\RR{p}}( (x, y ) ) \right\}  \\
 	&= \max\{ |y_q - y|, |y_p - y| \} = \max \{\xdistance_{qq}(y),\xdistance_{pp}(y)\}.
 	\end{align*}
 	Where the last equality follows from the observation in the proof of Lemma~\ref{lem:backward} that for any point $p$, the function $x \mapsto\|(x,y)-p\|$ is convex, and minimized at $x=x_p$.
 \end{proof}
 
 For all pairs of points $(p, q) \in \Skew$ either $(p,q)$ is a backward pair or $x_p < x_q$, and thus we obtain the following lemma.

\begin{lemma}
\label{lem:rewrite}
  For any polygonal curve $P$ and any $y$,
  \[
  \DB(y) = \max \left\{ \xdistance_{pq}(y) \mid (p,q)\in\mathcal{B}(P) \right\} = \max \{ \xdistance'_{p q}(y) \mid (p, q) \in \Skew \}.
  \]
\end{lemma}


%

\subsubsection{Relating $\DB(y)$ to furthest segment Voronoi diagrams}
\label{sub:relatingVoronoiDiagrams}

In this section we devise a divide and conquer algorithm
that computes $\DB(y)$ by computing it for subsets of 
vertices of $P$.
Lemma~\ref{lem:rewrite} allows us to express $\DB(y)$ in terms of $\Skew$ instead of $\mathcal{B}(P)$.
Next we refine the definition of $\DB(y)$ to make it decomposable.
To that end, we define $\DB(y)$ on pairs of subsets of $P$.
Let $S, T$ be any two subsets of vertices of $P$, we define:
\begin{equation*}
  \DB^{S \times T}(y) = \max \left\{ \xdistance'_{pq}(y) \mid (p,q) \in (S\times T) \cap \Skew \right\}.
\end{equation*}

We show that we can compute $\DB^{S \times T}(y)$
efficiently using the $\xdistance'$ functions. To this end, we fix a value of $y$ and show that computing $\DB^{S \times T}(y)$ is
equivalent to computing an intersection between two curves that consist of a linear number of pieces, each of constant complexity.
We then argue that as $y$ changes, the intersection point moves along a linear complexity curve that can be computed in $O(n\log n)$ time. 
This allows us to query $\DB(y)=\DB^{P \times P}(y)$ in $O(\log n)$ time, for any query height~$y$.

\subparagraph{From distance to intersections.} For a fixed value $y'$,
computing $\DB^{S \times T}(y')$ is equivalent to computing an intersection point between two curves:

\begin{lemma}\label{lem:functioneq}
	Let $y' \in \R$ be a fixed height, let $p$ be a point in $P$, and let
	$T$ be a subset of the vertices of $P[p,p_n]$. 
	The graphs of the functions
	$x \mapsto h_{\RR{p}}( (x, y') )$ and
	$x \mapsto h_{\LL{T}}( (x, y') )$ intersect at a single point
	$(x^*,y')$. Moreover,
	$ \DB^{ \{ p \} \times T }(y') =   h_{\LL{T}}( (x^*, y') ) =    h_{\RR{p}}( (x^*, y') )$
\end{lemma}

\begin{proof}
	Recall that $P[p,p_n]$ is the subcurve of $P$ from $p$ to $p_n$.
	In Observation~\ref{obs:monotonic} we noted that for all fixed $y'$, the function $x \mapsto h_{\LL{T}}( (x, y') )$ is monotonically increasing.
	Similarly, for any point $p$, the function  $x \mapsto h_{\RR{p}}( (x,
	y') )$ is monotonically decreasing. Meaning that the value $\min_x \max
	\left \{  h_{\RR{p}}( (x, y') ),  h_{\LL{T}}( (x, y')) \right \}$ is
	realized at $x^*$. Notice that $(x^*,y')$ is a unique point as we assume general position, i.e., no two points have the same $y$-coordinate. Next, we apply the definition of $\DB^{ \{ p \} \times T }(y')$:
	\begin{align*}
	\DB^{ \{ p \} \times T }(y') = 
	\max_{q \in T} \left\{ \delta'_{pq}(y') \right\} &=  
	\max_{q \in T} \left\{ \min_x \max \{ h_{\RR{p}}( (x, y') ), h_{\LL{q}}( (x, y') ) \} \right\} = \\
	\min_x \max \left \{  h_{\RR{p}}( (x, y') ) ,   \max_{q \in T} \{   h_{\LL{q}}( (x, y') ) \}   \right\} &= 
	\min_x \max \left \{  h_{\RR{p}}( (x, y') ),  h_{\LL{T}}( (x, y')) \right \} \Rightarrow \\
	\DB^{ \{ p \} \times T }(y') &= h_{\LL{T}}( (x^*, y') ) = h_{\RR{p}}( (x^*, y') ).
	\end{align*}
	\qedhere
\end{proof}

Lemma~\ref{lem:intersection} now follows easily from the previous.

\begin{lemma}\label{lem:intersection}
  Let $y' \in \R$ be a fixed height, and
	let 
	$S,T$ be subsets of vertices of $P$ such that the vertices in $S$
    precede all vertices of $T$.
	The graphs of the functions
	$x \mapsto h_{\RR{S}}( (x, y') )$ and
	$x \mapsto h_{\LL{T}}( (x, y') )$ intersect at a single point
	$(x^*,y')$. Moreover, $ \DB^{ S\times T }(y') = h_{\RR{S}}( (x^*, y') )  =  h_{\LL{T}}( (x^*, y') )  $.
\end{lemma}

\begin{proof}
	If all points in $S$ precede all points in $T$, then all elements in $S \times T$ are in $\Skew$ and we note: $\DB^{S \times T}(y') = \max_{p \in S} \left\{ \, \DB^{ \{p \} \times T }(y')\right\}.$ The equality then  follows from Lemma~\ref{lem:functioneq}.
\end{proof}

\begin{figure}[tb]
  \centering
  \includegraphics{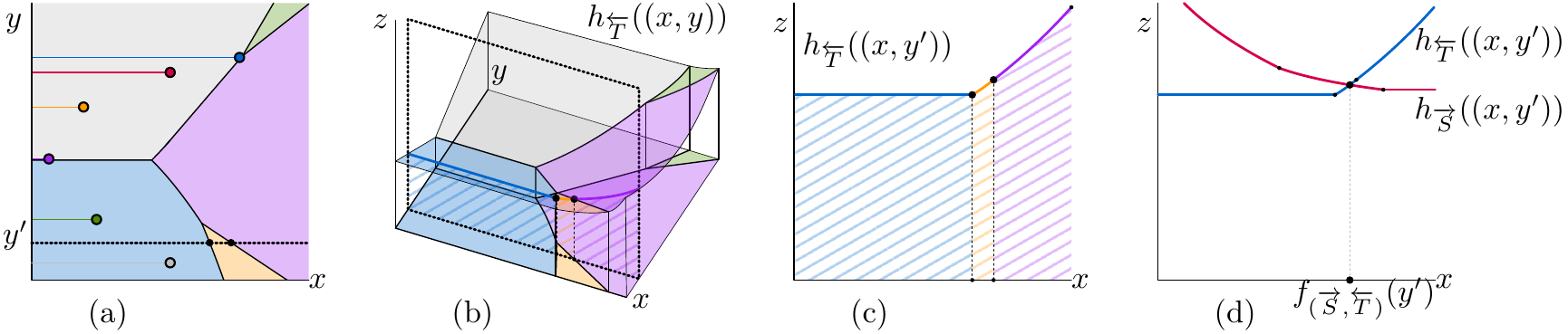}
  \caption{(a) A set $\LL{T}$ of rays arising from a set $T$ of points, with their FSVD.  
    (b) $h_{\LL{T}}( (x, y) )$ is the
    distance from $(x,y)$ to the ray corresponding to the Voronoi
    cell at $(x,y)$.  (c) For a fixed $y'$, 
     $x \mapsto h_{\LL{T}}( (x, y') )$ is monotonically
    increasing.  
    (d) The value $x$ for which
    $h_{\LL{T}}( (x, y') )=h_{\RR{S}}( (x, y') )$ corresponds to  $f_{({\RR{S},\LL{T}})}(y')$, and to $\DB^{S\times T}(y')$.}
  \label{fig:functionexample}
\end{figure}

Given such a pair $S, T$, for a fixed value $y'$, we can compute a
linear-size representation of $x \mapsto h_{\LL{T}}( (x, y') )$ in
$O(n \log n)$ time as follows (see Fig.~\ref{fig:functionexample}). We
compute the $\FSVD$ of $\LL{T}$ in
$O(n \log n)$ time.  Then, we compute the Voronoi cells intersected by
a line of height $y'$ (denoted by $\ell_{y'}$) in left-to-right order in
$O(n \log n)$ time. Suppose that a segment of $\ell_{y'}$ intersects only
the Voronoi cell belonging to a halfline $\LL{q} \in \LL{T}$, then on
this domain the function $h_{\LL{T}}( (x, y') ) = h_{\LL{q}}( (x, y') )$,
and thus it has constant complexity. A horizontal line can intersect at
most a linear number of Voronoi cells, hence the function has total
linear complexity.
Analogous arguments apply to $x \mapsto h_{\RR{S}}( (x, y') )$.


\subparagraph{Varying the $y$-coordinate.}
Let $f_{({\RR{S},\LL{T}})}\colon y\mapsto x^*$
be the function that for each $y$ gives the intersection point $x^*$ such that  $h_{\RR{S}}( (x^*, y) )  =  h_{\LL{T}}( (x^*, y) )$.
The intersection point $(x^*,y')$ lies on a Voronoi edge
of the $\FSVD$ of $(\LL{T}\cup \RR{S})$. More
precisely, it lies on the bichromatic bisector of the $\FSVD$ of
$\LL{T}$ and the one of $\RR{S}$ (see Fig.~\ref{fig:surface}).
When we vary the $y$-coordinate, the intersection point traces this bisector.
This implies that, given the $\FSVD$ of $\RR{S}$ and the $\FSVD$ of $\LL{T}$, the graph of $f_{({\RR{S},\LL{T}})}$ can be computed in $O(n)$ time.

\begin{figure}
	\centering
	\includegraphics{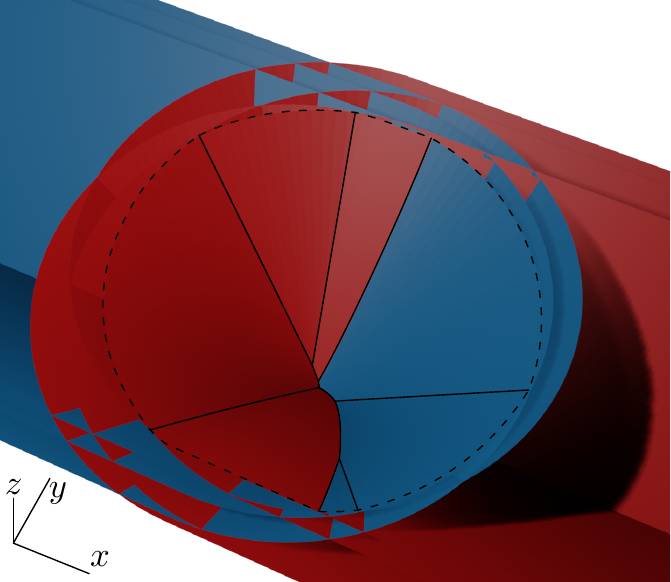}
	\caption{The $h_{\LL{T}}$ (blue) and $h_{\RR{S}}$ (red) functions
		(here clipped at $z=1$) intersect in a curve corresponding to
		$y \mapsto f_{(\LL{T},\RR{S})}(y)$.
	}
	\label{fig:surface}
\end{figure}

\begin{lemma}\label{lem:horizontal_cross_pairs_ds}
	Let $S, T$ be subsets of vertices of $P$ such that all vertices in $S$
	precede all vertices in~$T$. The function $\DB^{S \times T}$ has
	complexity $O(n)$ and can be computed in $O(n \log n)$
	time. Evaluating $\DB^{S \times T}(y)$, for some query value
	$y\in \R$, takes $O(\log n)$ time.
\end{lemma}
\begin{proof}
	We consider the function $y\mapsto (f_{({\RR{S},\LL{T}})}(y),\DB^{S \times T}(y))$. The graph of this curve is $y$-monotone. 
	Each $y$ has a unique value $f_{({\RR{S},\LL{T}})}(y)$, and we need to compute $\DB^{S \times T}(y))=h_{\RR{S}}( f_{({\RR{S},\LL{T}})}(y) , y) )  =  h_{\LL{T}}( f_{({\RR{S},\LL{T}})}(y), y) )$. 
	Hence the function $y\mapsto (f_{({\RR{S},\LL{T}})}(y),\DB^{S \times T}(y))$ is a well-defined curve in $\R^3$
	parameterized only by $y$. Thus, we consider the graph of the
	function $y\mapsto (f_{({\RR{S},\LL{T}})}(y),\DB^{S \times T}(y))$
	projected into the $(y,z)$-plane. It has linear complexity and can be
	computed in $O(n \log n)$ time. By storing the breakpoints of this
	function in a balanced binary search tree we can then evaluate
	$\DB^{S\times T}(y)$, for any $y$, in $O(\log n)$ time.
\end{proof}

\subsubsection{Applying divide and conquer}
\label{sub:divide-and-conquer}

We begin by analyzing the complexity of the function $\DB(y)$.
Consider a partition of $P$ into subcurves $S$ and $T$ with at most
$\lceil n/2 \rceil$ vertices each, and with $S$ occurring before $T$
along $P$. 
Our approach relies on the following fact.

\begin{observation}
  \label{obs:db_minx}
  Let $P$ be partitioned into two subcurves $S$ and $T$ with all
  vertices in $S$ occurring on $P$ before the vertices of $T$. We have
  that
  $
    \DB(y)=\DB^{P \times P}(y) = \max \left\{ \DB^{S \times S}(y), \DB^{S
    \times T}(y), \DB^{T \times T}(y) \right\}
    $.
\end{observation}

    
Note that we can omit
$\DB^{T \times S}$ because $(T\times S) \cap \Skew = \emptyset$.  

\begin{lemma}\label{lem:backward_pair_ds}
	Let $P$ be a polygonal curve with $n$ vertices. Function $\DB$
	has complexity $O(n\log n)$ and can be computed in $O(n\log^2 n)$
	time. Evaluating $\DB(y)$, for a given $y \in \mathbb{R}$, takes
	$O(\log n)$ time
\end{lemma}
\begin{proof}
	By Observation~\ref{obs:db_minx}, $ 
	\DB(y)=\DB^{P \times P}(y) = \max \left\{ \DB^{S \times S}(y), \DB^{S
		\times T}(y), \DB^{T \times T}(y) \right\}$, and by
	Lemma~\ref{lem:horizontal_cross_pairs_ds}, the complexity of
	$\DB^{S \times T}$ is $O(n)$. Hence, there are $O(n)$ backward pairs
	from $S \times T$ that could contribute to $\DB^{P \times P}$. Let
	$C(n)$ denote the number of backward pairs contributing to
	$\DB^{P \times P}$.  It follows that
	$C(n) = 2C(\lceil n/2 \rceil) + O(n)$, which solves to $O(n\log
	n)$. Since the complexity of $\DB = \DB^{P \times P}$ is
	linear in the number of contributing backward pairs~\cite{de2017data}, $\DB$ has
	complexity $O(n \log n)$.
	
	To compute $\DB$ we apply the same divide and conquer strategy. We recursively partition $P$ into roughly equal size subcurves
	$S$ and $T$. At each step, we compute the (graph of
	the) function $\DB^{S \times T}$, and merge it with the recursively
	computed functions $\DB^{S \times S}$ and $\DB^{T \times T}$. By
	Lemma~\ref{lem:horizontal_cross_pairs_ds}, computing
	$\DB^{S \times T}$ takes $O(n\log n)$ time. Computing the upper
	envelope of $\DB^{S \times T}$, $\DB^{S \times S}$, and
	$\DB^{T \times T}$, takes time linear in the complexity of the
	functions involved. The function $\DB^{S \times T}$ has complexity
	$O(n)$. However, $\DB^{S \times S}$, $\DB^{T \times T}$, and the output $\DB^{P \times P}$,  have complexity $O(n\log n)$. Hence,
	we spend $O(n\log n)$ time to compute $\DB^{P \times P}$. The total
	running time obeys the recurrence $R(n) = 2R(n/2) + O(n\log n)$, 
	that resolves to $O(n\log^2 n)$ time.
	
	We can easily store (the breakpoints of) $\DB^{P \times P}$ in a
	balanced binary search tree so that we can evaluate $\DB(y)$ for some
	query value $y$ in $O(\log (n \log n) ) = O(\log n)$ time.
\end{proof}
Eq.~\ref{eq:decompose_FD} together with Theorem~\ref{thm:dHtoPoint_algo} and Lemma~\ref{lem:backward_pair_ds} thus imply that
we can store $P$ in an $O(n\log n)$ size data structure so that we can
compute $\fd(P,\overline{ab})$ for some horizontal query segment
$\overline{ab}$ in $O(\log n)$ time. That is, we have established
Theorem~\ref{thm:horizontal_full_curve_ds}.

\subsection{ Querying for subcurves}
\label{sec:Horizontal:_Subtrajectories}

In this section we extend our data structure to support \frechet
distance queries to subcurves of $P$, establishing 
Theorem~\ref{thm:horizontal_subcurve_ds}. A query now consists of two
points $s$ and $t$ on $P$ and the horizontal query segment
$\overline{ab}$, and we wish to efficiently report the \frechet
distance $\fd(P[s,t],\overline{ab})$ between the subcurve $P[s,t]$,
from $s$ to $t$, and $\overline{ab}$. We show that we can support such
queries in $O(\log^3 n)$ time using $O(n\log^2 n)$ space.

We assume that given $s$ and $t$ we can determine the
edges of $P$ containing $s$ and $t$ respectively, in constant
time. 
Note that this is the case, for instance, when $s$ is given as a
pointer to its containing edge together with a location. If $s$ and
$t$ are given only as points in the plane, and $P$ is not
self-intersecting, we can find these edges in $O(\log n)$ time using a 
linear-size data structure for vertical ray-shooting~\cite{sarnak86planar_point_locat_using_persis_searc_trees}
on $P$. If $P$ does contain self-intersections this requires more
space and preprocessing time~\cite{agarwal_ray_1992}.

By Eq.~\ref{eq:decompose_FD}, $\fd(P[s,t],\overline{ab})$ can again be
decomposed into four terms, the first two of which can be trivially
computed in constant time. 
We build two separate data structures for the remaining two terms: the Hausdorff distance and the backward pair distance terms.
%

\subsubsection{Hausdorff distance for subcurves}
\label{sub:hausdorff-subtrajectories}

We build a data structure on $P$ such that given points $s,t$ on $P$
and $\overline{ab}$ we can report $\dhd(P[s,t],\overline{ab})$
efficiently. 
In particular, we use a two-level data structure of size $O(n \log n)$ that supports queries
in $O(\log^2 n)$ time, after $O(n \log n)$ preprocessing time, based on two observations: 

\begin{enumerate}
	\item The Hausdorff distance is decomposable in its first
	argument. That is:
	\[
	\textnormal{for all points } m \in P[s, t], \quad
	\dhd(P[s,t],\overline{ab}) =
	\max\left\{\dhd(P[s,m],\overline{ab}),\dhd(P[m,t],\overline{ab})\right\}.
	\]
	\item The Hausdorff distance $\dhd(P[s,t],\overline{ab})$ is realized
	by $s$, $t$, or a vertex $p$ of $P[s,t]$, that is:
	\[
	\textnormal{there is a vertex $p \in P$}   \textnormal{ s.t. }\dhd(P[s,t],\overline{ab}) =
	\max \{\dhd(p,\overline{ab}), \quad \dhd(s,\overline{ab}), \quad \dhd(t,\overline{ab})\}.
	\]
\end{enumerate}

\noindent
The data structure is a balanced binary search tree (essentially, a
1D-range tree) in which the leaves store the
vertices of $P$, in the order along $P$. Each internal node $\nu$
represents a \emph{canonical subcurve} $P_\nu$ and stores the vertices
of $P_\nu$ in the data structure of
Theorem~\ref{thm:dHtoPoint_algo}. Since these associated data
structures use linear space, the total space used is $O(n\log
n)$. Building the associated data structures from scratch would take
$O(n\log^2 n)$ time.  However, the following lemma immediately implies
that this time can be reduced to $O(n \log n)$:

\begin{figure}[t]
	\centering
	\includegraphics{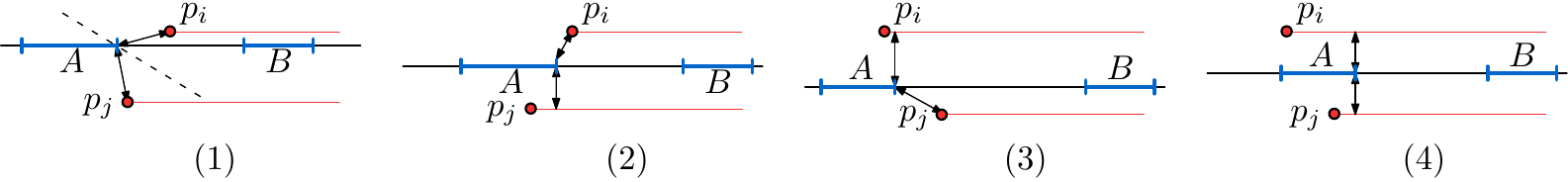}
	\caption{
		The four cases considered in the proof of Lemma~\ref{lemma:linearmerge} in order.
	}
	\label{fig:linearmerge}
\end{figure}

\begin{lemma}
	\label{lemma:linearmerge}
	Two instances of the data structure of Theorem~\ref{thm:dHtoPoint_algo}, consisting of a furthest segment Voronoi diagram preprocessed for point location, can be merged in linear time.
\end{lemma}

\begin{proof}
	Since two furthest segment Voronoi diagrams can be merged in linear time~\cite{PapadopoulouD13}, it remains to show that the associated point location data structure can also be constructed in linear time.
	This is a well-known result if the regions are  $y$-monotone~\cite{edelsbrunner_optimal_1986}. 
	
	We prove that the furthest segment Voronoi diagram of $\RR{S}$ has $y$-monotone cells (the case of  $\LL{S}$ is symmetric).
	
	Suppose for the sake of contradiction  that there is some horizontal line $\ell_y$ at height $y$, and a ray $\RR{p_i} \in \RR{S}$, such that the furthest Voronoi region of $\RR{p_i}$, intersected by $\ell_y$, has at least two maximal disjoint intervals $A$ and $B$  where $A$ is left of $B$ (these intervals contain the boundary of the associated Voronoi region).
	Consider the rightmost point $x$ in $A$.  Since $x$ coincides with the right boundary of $A$, there is a ray $\RR{p_j} \in \RR{S}$ such that $d((x,y), \RR{p_i} ) = d((x,y),\RR{p_j})$, and for an arbitrary small $\varepsilon > 0$, $d((x + \varepsilon,y), \RR{p_i} ) < d((x + \varepsilon,y),\RR{p_j})$.
	We make a distinction based on whether the distance from $(x, y)$ to $\RR{p_i}$ is realized by the distance  to the point $p_i$ or by the vertical distance to the line supporting $\RR{p_i}$  (similarly for the distance from $(x, y)$ to $\RR{p_j}$).
	Refer to Fig.~\ref{fig:linearmerge}.
	The first two cases show that the interval $B$ is empty and the last two show that the interval $A$ does not end at $x$.

	\paragraph{Case 1: $d((x,y), \RR{p_i} )  = d((x,y), p_i )$ and $d((x,y), \RR{p_j} )  = d((x,y), p_j )$.}
	In this case, $x$ lies on the bisector between $p_i$ and $p_j$ and the interval $A$ lies left of this bisector. Hence $p_i$ must lie right of $p_j$. 
	It follows that for all $x' > x$, $d( (x', y), \RR{p_i}) < d( (x', y), \RR{p_j})$. 
	Indeed if for such $x'$, $d ((x', y),  \RR{p_j} ) = d( (x', y), p_j)$ then it must be that  $d ((x', y),  \RR{p_i} ) = d( (x', y), p_i)$. Since $x'$ lies right of the bisector between $p_i$ and $p_j$ it follows that $d( (x', y), \RR{p_i}) = d( (x', y), p_i) < d( (x', y), \RR{p_j}) = d( (x', y), p_j)$.
	For all $x' > x$ where $d((x', y),  \RR{p_j} ) \neq d( (x', y), p_j)$, the value $d((x', y), \RR{p_j})$ remains constant whilst the value $d((x', y), \RR{p_i}))$ decreases or stays constant. This contradicts the assumption that the farthest Voronoi cell of $\RR{p_i}$ intersects $\ell_y$ right of $x$.


	\paragraph{Case 2: $d((x,y), \RR{p_i} )  = d((x,y), p_i )$  and $d((x,y), \RR{p_j} )  \neq d((x,y), p_j )$.}
	
	In this case, for all $x' > x$, $d((x',y), \RR{p_j}) = d((x,y), \RR{p_j} )$ and $d((x',y), \RR{p_i} ) < d((x,y), \RR{p_i} )$. This contradicts the assumption that the farthest Voronoi cell of $\RR{p_i}$ intersects $\ell_y$ right of $x$. 
	
	\paragraph{Case 3: $d((x,y), \RR{p_i} )  \neq d((x,y), p_i )$  and $d((x,y), \RR{p_j} )  = d((x,y), p_j )$}
	In this case, for all $x' > x$, $d( (x', y), \RR{p_j}) < d((x,y), \RR{p_j} )$ and $d((x',y), \RR{p_i} ) = d((x,y), \RR{p_i} )$ which contradicts the assumption that for an arbitrary small $\varepsilon > 0$, $d((x + \varepsilon,y), \RR{p_i} ) < d((x + \varepsilon,y),\RR{p_j})$.


	\paragraph{Case 4: $d((x,y), \RR{p_i} )  \neq d((x,y), p_i )$  and $d((x,y), \RR{p_j} )  \neq d((x,y), p_j )$}
	In this case, for all $x' > x$, $d((x',y), \RR{p_j} ) = d((x,y), \RR{p_j} )$ and $d((x',y), \RR{p_i} ) = d((x,y), \RR{p_i} )$
	which contradicts the assumption that for an arbitrary small $\varepsilon > 0$, $d((x + \varepsilon,y), \RR{p_i} ) < d((x + \varepsilon,y),\RR{p_j})$.
	%
\end{proof}

Let $s,t$ be two points on $P$, and let $s'$ and $t'$ be the first
vertex succeeding $s$ and preceding $t$, respectively.  The terms
$\dhd(P[s,s'],\overline{ab})$ and $\dhd(P[t',t],\overline{ab})$ can be
computed in $O(1)$ time.  There are $O(\log n)$ internal nodes in our
range tree, whose canonical subcurves together form $P[s', t']$ (see
e.g.~\cite[Chapter 5]{de1997computational}). For each such node $\nu$,
corresponding to a subcurve $P_\nu$, we query its associated data
structures in time logarithmic in the number of vertices in $P_\nu$ to
compute $\dhd(P_\nu,\overline{ab})$. We report the maximum distance
found and conclude:

\begin{lemma}
	\label{lem:hausdorff_term_subtrajectories}
	Let $P$ be a polygonal curve in $\R^2$ with $n$ vertices. In
	$O(n\log n)$ time we can construct a data structure of size
	$O(n\log n)$ so that given a horizontal query segment
	$\overline{ab}$, and two points $s,t$ on $P$,
	$\dhd(P[s,t],\overline{ab})$ can be computed in $O(\log^2 n)$ time.
\end{lemma}

Note that since we are given pointers to $s$ and $t$ we can make the
query time sensitive to the complexity $|P[s,t]|$ of the query
subcurve. Let $\nu_1,\dots,\nu_k$ be the nodes whose canonical subcurves
(ordered along $P$) make up $P[s',t']$. It is well known that a prefix
of $\nu_1,\dots,\nu_\ell$ are (a subset of the) right children on the
path connecting the leaf representing $s'$ to the lowest common
ancestor (LCA) $\mu$ of this leaf and the leaf representing $t'$,
whereas the remaining nodes $\nu_{\ell+1},\dots,\nu_k$ are left children
along the path from $\mu$ to the leaf representing $t'$. By
preprocessing our tree for $O(1)$ time LCA queries~\cite{rmq} and
storing a pointer from every internal node to its last leaf, we can
find these $k$ nodes in $O(k)$ time: given $\nu_i$, jump to the last
leaf in its subtree. This leaf represents some vertex $p_j$. We can
then compute the LCA of $p_{j+1}$ and $\nu_i$. The next node
$\nu_{i+1}$ is the right child of the node found. Moreover, since the
sizes of successive nodes in $\nu_1,\dots,\nu_\ell$ (at least) double
there are at most $\ell = O(\log |P[s,t]|)$ such nodes. Symmetrically
$k-\ell = O(\log |P[s,t]|)$, hence $k = O(\log |P[s,t]|)$. Since each
of the canonical subcurves has size at most $|P[s,t]|$, querying their
associated data structures also takes at most $O(\log |P[s,t]|)$
time. Hence, our queries take only $O(\log^2 |P[s,t]|)$ time in total.

\subsubsection{Backward pair distance for subcurves}
\label{sub:backward_pair_subcurve_queries}

In this section we describe how to store $P$ so that given points $s$
and $t$ on $P$ and the horizontal query segment $\overline{ab}$ we can
compute the backward pair distance $\DB^{P[s,t]\times P[s,t]}(y_a)$
efficiently.
The main idea to
support subcurve queries is to store all intermediate results of the
divide and conquer algorithm from
Section~\ref{sub:divide-and-conquer}. Hence, our main data structure
is a 1D-range tree whose leaves store the vertices of $P$, ordered
along $P$. Each internal node $\nu$ corresponds to some subcurve
$P_\nu$ of $P$, and will store the function
$\DB^{P_\nu \times P_\nu}(y)$ (Lemma~\ref{lem:backward_pair_ds}) as
well as the functions $h_{\LL{P_\nu}}(x, y)$ and
$h_{\RR{P_\nu}}(x, y)$ represented by two new data structures that we
will denote by $\Delta_{\LL{P_\nu}}(y)$ and $\Delta_{\RR{P_\nu}}(y)$.

We first sketch our query approach, in the following subsections we show how to construct the specific data structure and queries.
We are
given $y'$, and $s,t \in P$. We then identify the vertex $s'$ succeeding
$s$ and the vertex $t'$ preceding $t$, and define $P_0=P[s,s']$ and $P_{k+1}=P[t',t]$. 
There are $k=O(\log n)$ internal nodes whose canonical
subcurves $P_1, P_2,\dots$ together form $P[s', t']$.
We prove that $\DB^{P[s, t] \times P[s, t]}(y')$ is the maximum over two terms:
\[
\DB^{P_i \times P_i}(y'), \text{ for all $i$, and }\quad
\DB^{P_i \times P_j}(y') \text{ for $i < j$.}
\]
\noindent
The $O(\log n)$ values of the first term
can be computed in $O(\log^2 n)$ total time, using the  data
structures of Lemma~\ref{lem:backward_pair_ds} stored in each node. The second
term contains $O(\log^2 n)$ values, and we show how to compute each
value in $O(\log n)$ time, using the new data structures, for a total
query time of $O(\log^3 n)$.
This query time then dominates the time
it takes to compute the maximum of all these terms.

\subsubsection{Using furthest segment Voronoi diagrams to compute $\DB^{S \times T}(y')$}
\label{sub:cross_terms}

Let $S,T$ be two contiguous subcurves of $P$, with $S$ occurring strictly
before $T$ on $P$.
We first study how to compute $\DB^{S \times T}(y')$ efficiently for any given $y'$.
For ease of exposition, we first show how to construct a linear-size
data structure on $T$ such that given a query point $p$ that precedes $T$ along $P$ and a value $y'$, we can compute
$\DB^{\{p\} \times T}(y')$ in $O(\log |T|)$ time. By
Lemma~\ref{lem:functioneq} this amounts to computing the point of
intersection between the functions $x \mapsto h_{\RR{p}}( (x , y') )$
and $x \mapsto h_{\LL{T}}( (x, y') )$. 


Our global approach is illustrated in
Fig.~\ref{fig:functionexample}. Suppose for the ease of exposition,
that $y'$ is fixed. Given the query $p$,  we can compute the
convex function $x \mapsto h_{\RR{p}}( (x, y') )$ in constant time. For
the set $\LL{T}$, we note that Observation~\ref{obs:fsvd} shows that
the (graph of the) function $(x, y) \mapsto h_{\LL{T}}( (x, y) )$ is
an upper envelope whose maximization diagram is the furthest segment
Voronoi diagram of the set of halflines $\LL{T}$. For any fixed $y'$
and $\LL{T}$, we can construct a data structure $\Delta_{\LL{T}}(y')$ that
stores the Voronoi cell edges that are intersected by a horizontal
line of height $y'$ in their left-to-right order. Specifically, we
store a red-black tree where each node stores a consecutive pair of
edges (in this way, a node in the tree corresponds to a unique Voronoi
cell in the diagram). Since a furthest segment Voronoi diagram has a
linear number of edges, $\Delta_{\LL{T}}(y')$ has $O(|T|)$ size.

\begin{lemma}
	Let $p \in P$ be a query point, let $y'$ be some given value, and let $T$ be a sequence of points of $P$ that succeed $p$.
	Given $\Delta_{\LL{T}}(y')$ we can compute the point $(x^*, y')$ where $h_{\RR{p}}( (x^*, y') ) = h_{\LL{T}}( (x^*, y') )$ in $O(\log |T|)$ time. 
\end{lemma}

\begin{proof}
	Consider the two edges stored at the root of $\Delta_{\LL{T}}(y')$. 
	These two edges partially bound a Voronoi cell corresponding to a halfline $\LL{r} \in \LL{T}$. 
	In constant time, we can compute the point of intersection $(x', y')$ between $x \mapsto h_{\RR{p}}( (x, y') )$ and $x \mapsto h_{\LL{r}}( (x, y') )$.
	
	Using the two edges, we can compute the $x$-interval of the Voronoi cell of $\LL{T}$ that intersects the horizontal line of height $y'$ in constant time. 
	If the value $x'$ lies in this interval, then we have found the unique point of intersection between $h_{\RR{p}}$ and $h_{\LL{T}}$ (in other words, $(x', y') = (x^*, y')$). 
	If $x'$ lies left of this $x$-interval, we can disregard all nodes in the right subtree of the root. This is because for all $x \geq x'$, 
	$h_{\LL{T}}(x,y') \geq h_{\LL{r}}(x,y') > h_{\RR{p}}(x,y')$.
	A symmetrical property holds  when $x'$ lies right of this $x$-interval. 
	Hence at every step, we can discard at least half of the remaining candidate halflines in $\LL{T}$ and thus we can compute the halfline in $\LL{T}$ that forms the intersection with $x \mapsto h_{\RR{p}}(x,y')$ in logarithmic time. 
\end{proof}

Now, we extend our approach from products of points and sets, to products between sets.
Recall that for any $\RR{S}$, the function $x \mapsto h_{\RR{S}}(x, y')$ is  piecewise monotone, where each piece is a constant-complexity segment coinciding with $x \mapsto h_{\RR{s}}(x, y')$ for some $\RR{s} \in \RR{S}$.

\begin{lemma}
	\label{lem:oracle}
	Let $S, T \subset P$ such that all points in $S$ precede all points in
	$T$ and let $y'$ be some given value. Let $\sigma_1, \sigma_2$ be two points
	on $x \mapsto h_{\RR{S}}(x, y')$ that bound a curve
	segment with $\sigma_1$ left of $\sigma_2$. 
	Moreover, let $\tau_1, \tau_2$ be defined analogously.
	If the segments bounded by $(\sigma_1, \sigma_2)$ and $(\tau_1, \tau_2)$ do not intersect then the intersection between $h_{\RR{S}}(x, y')$ and $h_{\LL{T}}(x, y')$ lies either: left of $\sigma_1$, right of $\sigma_2$, left of $\tau_1$, or right of $\tau_2$, and we can identify the case in $O(1)$ time.
\end{lemma}

\begin{figure}[tb]
	\centering
	\includegraphics{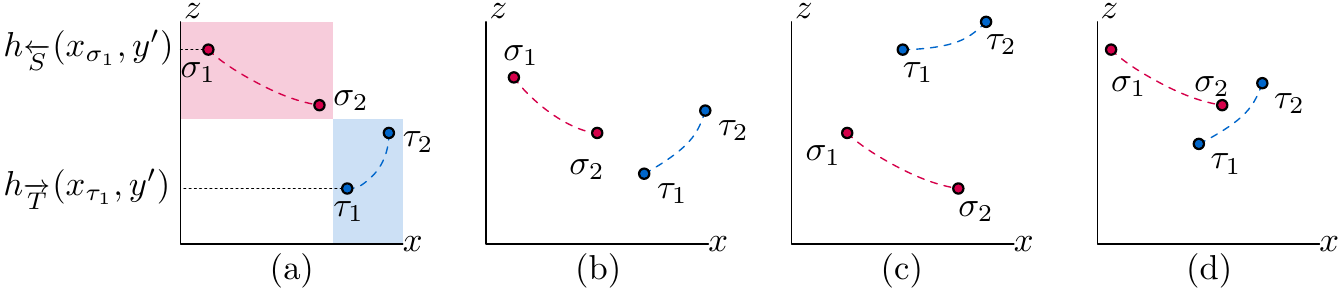}
	\caption{
		Cases in the proof of Lemma~\ref{lem:doublesearch}. In (d) both case 2 and 3 apply.
	}
	\label{fig:cases}
\end{figure}

\begin{proof}
	Recall that $x \mapsto h_{\RR{S}}((x, y'))$ and
	$x \mapsto h_{\LL{T}}((x, y'))$ are monotonically decreasing and
	increasing functions, respectively.
	The proof is a case distinction, illustrated by Fig.~\ref{fig:cases}.
	
	\subparagraph{Case 1: there is no vertical or horizontal line that
		intersects both segments.}  In this case, there is a separation of
	the $(x,z)$-plane into four quadrants, such that one segment lies in a top
	quadrant and the other segment in the opposite bottom
	quadrant. Let the segment of $h_{\RR{S}}$ lie in the top
	left quadrant, then all points on $h_{\RR{S}}$ left of this segment,
	are higher than $\sigma_1$. All points left of
	$h_{\RR{T}}$ must be lower than $\tau_1$. Hence all segments of $h_{\RR{S}}$ left $\sigma_1$
	cannot form the point of intersection and can be discarded. 
	Given $(\sigma_1, \sigma_2, \tau_1, \tau_2)$ we can identify this case in constant time.
	The three remaining (sub-)cases are symmetrical. 
	
	\subparagraph{Case 2: At least one horizontal line intersects both segments.}
	If the segments between $(\sigma_1, \sigma_2)$ and $(\tau_1, \tau_2)$ do not intersect, all horizontal lines intersect these two segments in the same order.
	Via the same argument as above we note that if the segment of $h_{\RR{S}}$ is intersected first, then all segments on $h_{\RR{S}}$ left of $\sigma_1$ can be discarded. Similarly then all segments of $h_{\LL{T}}$ right of $\tau_2$ can be discarded. The argument for when the order is reversed is symmetrical and we can compute the intersections with a horizontal line in $O(1)$ time.
	
	\subparagraph{Case 3: At least one vertical line intersects both segments. }
	In this case, each vertical line that intersects both segments must intersect them in the same order. We note that if the segment between $(\sigma_1, \sigma_2)$ is intersected first, then all segments on $h_{\RR{S}}$ right of $\sigma_2$ can be discarded, along with all segments right of $\tau_2$.
	The argument for when the order is reversed is symmetrical
	and we can compute the intersections with a vertical line in $O(1)$ time.
\end{proof}

\begin{lemma}
	\label{lem:doublesearch}
	Let $S, T \subset P$, all points in $S$ precede all points
	in $T$ and let $y'$ be some given value. Given $\Delta_{\RR{S}}(y)$ and $\Delta_{\LL{T}}(y')$, we can
	compute $\DB^{S \times T}(y)$ in $O(\log |S| + \log |T|)$ time.
\end{lemma}

\begin{proof}
	
	Consider the Voronoi edges stored at the root of
	$\Delta_{\RR{S}}(y')$, these are intersected by a horizontal line of height $y'$ in the $x$-coordinates $x_1$ and $x_2$ respectively.
	The edges of the Voronoi cell point to the cell $\RR{s} \in \RR{S}$ where we know that for all $x'$ with $x_1 < x' < x_2$, $h_{\RR{S}}(x', y') = h_{\RR{s}}(x', y')$.
	Given the pointer to $\RR{s} \in \RR{S}$ we compute the points $\sigma_1 = h_{\RR{S}}(x_1, y')= h_{\RR{s}}(x_1, y')$ and $\sigma_2 = h_{\RR{S}}(x_1, y') = h_{\RR{s}}(x_1, y')$ in constant time and the segment of $h_{\RR{S}}(x, y')$ connecting them.
	Similarly, we
	can compute a point $\tau_1$ and point $\tau_2$ for the edges stored at the root of
	$\Delta_{\LL{T}}(y')$, and the segment of $h_{\LL{T}}(x, y')$ connecting $(\tau_1, \tau_2)$.
	
	Then we can apply Lemma~\ref{lem:oracle} to the  pair  $(\sigma_1, \sigma_2)$, $(\tau_1, \tau_2)$, to discard at least all segments left/right of $(\sigma_1, \sigma_2)$ or all segments left/right of $(\tau_1, \tau_2)$. 
	Repeating this recursively on the remaining subtrees of $\Delta_{\RR{S}}(y')$ and $\Delta_{\LL{T}}(y')$, we can find the two segments defining the point of intersection in  $O(\log |S| + \log |T|)$ time.
\end{proof}

All that remains is to show that we can apply
Lemma~\ref{lem:doublesearch} not only to a given $y'$, but to any $y$.
To this end we sweep the plane with a horizontal line at height $y$
while maintaining $\Delta_{\LL{T}}(y)$ as a partially persistent red
black
tree~\cite{sarnak86planar_point_locat_using_persis_searc_trees}. We do
the same for $\Delta_{\RR{S}}(y)$.

\begin{lemma}\label{lem:persistent_sweep_ds}
	Let $S$ and $T$ be two sets of vertices of $P$.
	We can store $S$ in a data structure of size $O(|S|)$, and
	$T$ in a data structure of size $O(|T|)$ such that,
	if all vertices in $S$ precede all vertices
	in $T$, we can compute $\DB^{ S \times {T}}(y')$ for any fixed
	$y' \in \mathbb{R}$ in $O(\log (|S| + |T|))$ time.
	Building these data structures takes
	$O(|S|\log|S|)$ and $O(|T|\log|T|)$ time, respectively.
\end{lemma}

\begin{proof}
	Consider a horizontal sweepline at height $y'$, and consider the
	continuously changing data structure $\Delta_{\LL{T}}(y')$.  For any
	set $\LL{T}$, the FSVD has linear
	complexity~\cite{PapadopoulouD13}. Hence, for any $y'$,
	$\Delta_{\LL{T}}(y')$ contains at most a $O(|T|)$ edges. Moreover,
	there are only $O(|T|)$ $y$-coordinates at which the combinatorial
	structure of $\Delta_{LL{T}}(y)$ changes. At each such an event we
	make a constant number of updates to $\Delta_{\LL{T}}$. Since
	$\Delta_{\LL{T}}$ is a (partially persistent) red black tree, these
	changes take $O(\log |T|)$ time, and $O(1)$
	space~\cite{sarnak86planar_point_locat_using_persis_searc_trees}. We
	use the same preprocessing for $\RR{S}$. Finally, we observe that we
	can obtain $\Delta_{\LL{T}}(y')$ and $\Delta_{\RR{S}}(y')$ for a
	given $y' \in \R$ in $O(\log |T| + \log |S|)$ time.
\end{proof}
\subsubsection{Divide and conquer for subcurves} 
We denote by $\Delta_{\LL{T}}(y)$ the partially persistent red-black
tree on the furthest-segment Voronoi diagram of $\LL{T}$.  We now show
how for any two vertices $s, t \in P$ and height $y'$, we can compute
$\DB^{P[s, t] \times P[s, t]}(y')$ efficiently through repeated
application of Observation~\ref{obs:intersection_finding_decomposable}
and Lemma~\ref{lem:persistent_sweep_ds}.

\begin{observation}
	\label{obs:intersection_finding_decomposable}
	Let $S$ and $T$ be two sets of vertices of $P$ such that all points in $S$ precede all points in
	$T$ and $S = S' \cup S''$ and $T=T'\cup T''$. Then for all $y \in \R$: 
	\[
	\DB^{S \times T}(y) = \max \left\{ \DB^{S' \times T}(y), \DB^{S'' \times
		T}(y) \right\}
	\text{ and }
	\DB^{S \times T}(y) = \max \left\{ \DB^{S \times T'}(y), \DB^{S \times
		T''}(y) \right\}.
	\]
\end{observation}

\begin{proof}
	Recall that
	$\DB^{S \times T}(y) = \max \left\{ \xdistance'_{pq}(y) \mid (p,q)
	\in (S\times T) \cap \Skew \right\}$. Using that $S = S' \cup S''$
	we then get $S \times T = (S' \times T) \cup (S'' \times T)$, and
	thus also
	$(S\times T) \cap \Skew = ((S' \times T) \cap \Skew) \cup ((S''
	\times T) \cap \Skew)$. Since computing a maximum is decomposable we
	therefore get
	$\DB^{S \times T}(y) = \max\left\{\DB^{S' \times T}(y),\DB^{S''
		\times T}(y)\right\}$. Analogously, $\DB^{S \times T}(y) = \max \left\{ \DB^{S \times T'}(y), \DB^{S
		\times T''}(y) \right\}$.
\end{proof}

\begin{lemma}\label{lem:bisectorSubDS}
	Let $P$ be a polygonal curve in $\R^2$ with $n$ vertices. 
	We can build an $O(n\log^2 n)$ size data structure in $O(n\log^2 n)$ time, such that given any query $(s,t,y)$, for $s, t \in P$ (not necessarily vertices) we can
	compute $\DB^{P[s,t]\times P[s,t]}(y)$ in $O(\log^3n)$
	time.
\end{lemma}

\begin{proof}
	Our main data structure is a range tree whose leaves store the vertices of $P$, ordered
	along $P$. 
	Each internal node $\nu$ corresponds to some
	subcurve $P_\nu = P[p, q]$ for two vertices $p, q \in P$. We assume that $P_\nu$ consists of $m$ vertices.
	Each node will store the representation of function
	$\DB^{P_\nu \times P_\nu}(y)$ as an upper envelope. By Lemma~\ref{lem:backward_pair_ds}, this envelope has complexity $O(m \log m)$. The algorithm to compute this envelope is a divide and conquer algorithm whose recursion tree matches our range tree. Hence, when we compute $\DB^{P_\nu \times P_\nu}(y)$ for the root of our tree in $O(n\log^2 n)$ time we actually also construct all the envelopes associated with the other nodes. In addition, each node will store 
	$\Delta_{\LL{P_\nu}}(y)$ and $\Delta_{\RR{P_\nu}}(y)$
	which require $O(m)$ space and can be constructed by a divide and conquer approach in $O(m \log m)$ time. 
	Hence, the resulting data structure requires $O(n \log^2 n)$ space and can be constructed in $O(n \log^2 n)$ total time.
	
	Let $s,t$ be two points on $P$, and $s'$ and $t'$ be the first vertex
	succeeding $s$ and preceding~$t$, resp. There are $O(\log n)$ subtrees in
	our range search tree, whose subcurves $P[p,q]$ together form
	$P[s', t']$. We can identify the nodes bounding these subtrees in
	$O(\log n)$ time.  We denote these nodes by $P_1,..,P_k$, where
	the index matches the order along $P$. Furthermore, we define $P_0=P[s,s']$ and $P_{k+1}=P[t',t]$
	and observe:
	\begin{equation}
	\label{eq:decomp}
	\DB^{P[s,t]\times P[s,t]}(y)=\max_{0 \leq i \leq j \leq k+1}  
	\DB^{P_i\times P_j}(y) 
	\end{equation}
	\noindent
	Indeed, by Observation~\ref{obs:intersection_finding_decomposable},
	$\DB^{P[s,t]\times P[s,t]}(y)$ is decomposable. By
	repeated application of this observation we have that
	$\DB^{P[s,t]\times P[s,t]}(y) = \max_{i,j} \DB^{P_i\times
		P_j}(y)$. 
	Moreover, we only need to consider pairs with $i \leq j$,
	since if $j < i$, we have $P_i \times P_j \not\subseteq \Skew$.
	Thus, we can compute $\DB^{P[s, t] \times P[s, t] }(y')$ by only computing the values $\DB^{P_i \times P_j }(y')$ with $i \le j$. 
	
	By construction, for each of the $O(\log^2 n)$ pairs $S, T$ in our
	decomposition (Eq.~\ref{eq:decomp}) it holds that all points in $S$ precede
	all points in $T$. Thus we can compute $\DB^{S \times T}(y')$ in $O(\log (|S|+|T|) )=O(\log n)$ time by Lemma~\ref{lem:persistent_sweep_ds}.
	Computing this value for each pair takes $O(\log^3 n)$ total
	time.  For each subcurve $P_i$, we compute $\DB^{P_i \times P_i}(y)$ in $O(\log n)$ time
	using the data structure for $\DB^{P_\nu \times P_\nu}(y)$ of Lemma~\ref{lem:backward_pair_ds} stored in
	the node $\nu$ corresponding to $P_i$ (or in $O(1)$ time if $i \in \{0,k+1\}$).
	The lemma follows by taking the maximum of these $O(\log^2 n)$ values.
\end{proof}

As with the Hausdorff term we can make the query time sensitive to
the complexity of $P[s,t]$, and obtain a query time of
$O(\log^3 |P[s,t]|)$. Theorem~\ref{thm:horizontal_subcurve_ds} now
follows.

\section{Arbitrary orientation queries}
\label{sec:Arbitrary_Query_Orientation}

In this section we extend our results to arbitrarily
oriented query segments, proving
Theorem~\ref{thm:arbitrary_full_curve_ds}. Let $\alpha$ be the slope
of the line containing the query segment $\overline{ab}$, and let $\beta$ be its
intercept (note that vertical query segments can be handled by
a rotated version of our data structure for horizontal queries).
We again consider the case where $a$ is left of $b$; the other case
is symmetric. 
Following Eq.~\ref{eq:decompose_FD}, we can write
$\fd(P,\overline{ab})$ as the maximum of four terms: $\|p_1-a\|$,
$\|p_n-b\|$, $\dhd(P,\overline{ab})$, and the \emph{backward pair
  distance} $\DB(\alpha,\beta)$ with respect to $\alpha$. The
backward pair distance is now defined~as
\begin{align*}
  \DB(\alpha,\beta) &= \max\left\{ \xdistance_{pq}(\alpha,\beta) \mid
                      (p,q)\in\mathcal{B}(P)\right\}, & \text{ where } \\
  \xdistance_{pq}(\alpha,\beta) &= \min_x \max \left\{\|(x,\alpha x +
  \beta)-p\|,\|(x,\alpha x + \beta)-q\| \right\}.
\end{align*}

In
Section~\ref{sub:arb_orientations_hausdorff} we present an $O(n\log n)$ size data structure that supports querying
$\dhd(P,\overline{ab})$ in $O(\log^2 n)$ time. 
The key insight is that we can 
use furthest \emph{point} Voronoi diagrams instead of furthest \emph{segment}
Voronoi diagrams. In Section~\ref{sub:arb_orient_backward_pair} we present a
data structure that efficiently supports querying
$\DB(\alpha,\beta)$. In
Section~\ref{sec:arb_Subcurve_Queries} we extend our results to
support queries against subcurves of $P$ as well. This combines our
insights from the horizontal queries with our results from
Sections~\ref{sub:arb_orientations_hausdorff}
and~\ref{sub:arb_orient_backward_pair}. Finally, in Section~\ref{sub:Space_vs_Query_time_tradeoff} we then show how this also leads to a space-time trade off.

\subsection{The Hausdorff distance term}
\label{sub:arb_orientations_hausdorff}

For any point $p$ and slope $\alpha$ we denote by $\LLa{p}$ the ray with slope $\alpha$ that points in the leftward direction. 
Similarly, for any point set
$T$, we define $\LLa{T} = \left\{ \LLa{p} \mid p \in T
\right\}$. Furthermore, we define $h_{\LLa{p}}(x, y)$ to be the
directed Hausdorff distance from $(x,y)$ to the ray $\LLa{p}$, and
$h_{ \LLa{T} } (x, y) = \max \{ h_{\LLa{p}}(x, y) \mid p \in T \}$.
Let $\CH(T)$ be the convex hull of $T$. 
For a given slope $\alpha$ and a point $r$, we denote by $\ell_r$
the line that is perpendicular to a line with slope $\alpha$ that goes through $r$.
Observe that $\LLa{\ell_r}$ is the halfplane to the left of $\ell_r$.

\begin{lemma}
	\label{lem:convexhalf}
	Let $\LLa{T}$ be a collection of halflines containing $\LLa{p}$ as a
	topmost and $\LLa{q}$ as a bottommost halfline,
	respectively (always with respect to $\alpha$). We have that
	\[
	h_{\LLa{T} } (x, y) = \max  \{ h_{\LLa{p}}(x, y),
	h_{\LLa{q}}(x, y),
	\max \{ ||s - (x, y) || \mid s \in \CH(T) \cap \LLa{\ell_{(x,y)}} \} \}
	\] 
\end{lemma}

\begin{proof}
	Let $t \in T$ be the point realizing
	$h_{\LLa{T} } (x, y) = \max \{ h_{\LLa{s}}(x, y) \mid s \in T \}$,
	for some query point $(x,y)$. Without loss of generality, we can
	assume $\alpha = 0$, and hence the line $\ell_{(x,y)}$ is
	vertical. 
	We distinguish two cases depending on the position of $t$.
	
	\begin{description}
		\item[Case $t$ right of $\ell_{(x,y)}$.] 
		In this case, $h_{\LLa{t}}(x,y)$  must be given by the vertical distance $|y_t-y|$, therefore we must have that $|y_t-y| = \max\{|y_p-y|,|y_q-y|\}$.
		Thus $\LLa{t}$ must be a topmost or bottommost halfline.
		
		\item[Case $t$ left of $\ell_{(x,y)}$.] Observe that for all points
		$u \in T$ left of $\ell_{(x,y)}$ we have that
		$h_{\LLa{u}}(x,y)=\|u-(x,y)\|$. Now assume by contradiction that
		$t \not\in \{p,q\}\cup ( \CH(T) \cap \LLa{\ell_{(x,y)}})$, and let
		$\overline{cd}$ be the (leftmost) edge of
		$\CH(T) \cap \LLa{\ell_{(x,y)}}$ hit by $\LLa{t}$. Since $t$
		realizes $h_{\LLa{T}}(x,y)$ we claim that $c$ lies inside the disk
		$D$ centered at $(x,y)$ that has $t$ on its boundary. If $c \in T$
		and outside $D$ this would immediately contradict that $t$
		realizes $h_{\LLa{T}}(x,y)$. If $c$ is the intersection point of
		an edge of $\CH(T)$ with $\ell_{(x,y)}$ and lies outside $D$ the
		endpoint $c'$ of this edge that lies right of $\ell_{(x,y)}$ would
		have
		$h_{\LLa{c'}}(x,y) > \|c-(x,y)\| > \|t-(x,y)\|=h_{\LL{t}}(x,y)$,
		contradicting that $t$ realizes $h_{\LLa{T}}(x,y)$. Hence, $c$
		lies inside $D$.
		
		Via the same argument as above $d$ lies inside $D$. So, by
		convexity, the line segment $\overline{cd}$ is completely
		contained in this disk $D$. However, since (by definition of
		$h_{\LLa{t}}(x,y)$) $t$ is the point on $\LLa{t}$ closest to
		$(x,y)$, the remainder of this halfline is outside $D$. Hence,
		$\LLa{t}$ does not intersect $\overline{cd}$. Contradiction.\qedhere
	\end{description}
\end{proof}

Our data structure will store $\CH(T)$ in a 1D-range tree whose
internal nodes store FPVDs. This allows us to evaluate
$h_{\LLa{T}}(x,y)$ and $h_{\RRa{T}}(x,y)$ for a query point $(x,y)$
and slope $\alpha$.

\begin{lemma}
	\label{lem:querying_h_arbitrary_orient}
	Let $T$ be a set of $n$ points in $\R^2$. In $O(n\log n)$ time we
	can construct a data structure of size $O(n\log n)$ so that given a
	query point $(x,y)$ and query slope $\alpha$ we can compute
	$h_{\LLa{T}}(x,y)$ in $O(\log^2 n)$ time.
\end{lemma}

\begin{proof}
	Consider a clockwise traversal of the convex hull of $T$ that visits
	every vertex twice. Let $t_1,\dots,t_{2k}$ denote the vertices in this
	order (so $t_{i+k}=t_i$). We store these vertices 
	$t_1,\dots,t_{2k}$ in the leaves of a range tree. Each internal node
	$\nu$ corresponds to some contiguous subsequence $T_\nu = t_i,\dots,t_j$
	of these vertices, and stores the furthest point Voronoi diagram (FPVD) of $T_\nu$. 
	Since the FPVD has
	linear size and the points are in convex position, it can be computed in linear time~\cite{AggarwalGSS89}.
	Thus our data structure has
	size $O(n\log n)$ and can be computed in $O(n\log n)$ time.
	
	To answer a query, we first find the bottom- and topmost points of
	$\CH(T)$ (and thus of $T$) with respect to slope $\alpha$. Let $q$ and
	$p$ be these points, respectively. We now find a contiguous
	subsequence $t_i,\dots,t_j$ of the vertices of $\CH(T)$ in the halfplane
	$\LLa{\ell_{(x,y)}}$. Note that since $t_1,\dots,t_{2k}$ traverses
	$\CH(T)$ twice such a contiguous sequence exists.  More specifically,
	we only compute the first vertex $t_i$ and the last vertex $t_j$.  We
	then query the data structure to obtain $O(\log n)$ internal nodes
	$\nu$ whose associated sets $T_\nu$ together represent $t_i,\dots,t_j$,
	and query their furthest point Voronoi diagrams to find the point $s$
	in $t_i,\dots,t_j$ furthest from $(x,y)$. We report the maximum of
	$h_{\LLa{p}}(x, y)$, $h_{\LLa{q}}(x, y)$, and $\|s-(x,y)\|$. By
	Lemma~\ref{lem:convexhalf} this is $h_{\LLa{T}}(x,y)$.
	
	Finding $p$, $q$, $t_i$, and $t_j$ takes $O(\log n)$ time. This is
	dominated by the $O(\log^2 n)$ time to query all FPVDs. Thus, we can
	compute $h_{\LLa{T}}(x,y)$ in $O(\log^2 n)$ time. Symmetrically, we
	can query $h_{\RRa{T}}(x,y)$ in $O(\log^2 n)$ time.
\end{proof}

By using similar ideas to those of range minimum queries~\cite{halfplane_queries,rmq}
we can achieve $O(\log n)$ query time using $O(n^2)$ space. Then
following Corollary~\ref{lem:dHtoPoint} we can also compute
$\dhd(P,\overline{ab})$ by querying either of these data structures
with points $a$ and $b$. Hence, we obtain:

\begin{lemma}\label{lem:querying_h_arbitrary_orient_large_space}
	Let $T$ be a set of $n$ points in $\R^2$. In $O(n^2)$ time we
	can construct a data structure of size $O(n^2)$ so that given a
	query point $(x,y)$ and query slope $\alpha$ we can compute
	$h_{\LL{T}}(x,y)$ in $O(\log n)$ time.
\end{lemma}

\begin{proof}
	For every vertex $t_i$ on $\CH(T)$, and every $\ell \in 1,\dots,\log n$
	we store the FPVD of $t_i,..t_{i+2^{\ell}}$. This takes a total of
	$\sum_{\ell=1}^{\log n} O(2^\ell) = O(2^{\log n})=O(n)$ space per
	vertex, and thus $O(n^2)$ space in total. Building the FPVDs takes
	linear time, so the total construction time is $O(n^2)$ as well.
	
	To evaluate $h_{\LL{T}}(x,y)$ for some query point $(x,y)$ and query
	slope $\alpha$ we again have to find the furthest point among some
	interval $t_i,\dots,t_j$ along $\CH(T)$. We compute the largest value
	$\ell$ such that $2^\ell \leq j-i$. This allows us to decompose the
	``query range'' $t_i,\dots,t_j$ into two overlapping intervals
	$t_i,\dots,t_{i+2^\ell}$ and $t_{j-2^\ell},\dots,t_j$ for which we have
	pre-stored the FPVD. We can thus query both these FPVDs, in
	$O(\log n)$ time, and report the furthest point found. Computing
	$\ell$ and finding the points $t_{i+2^\ell}$ and
	$t_{j-2^\ell}$ takes $O(\log n)$ time as well, by a binary search on
	$t_i,...,t_j$. Hence the total query time is $O(\log n)$.
\end{proof}

\begin{theorem}
  \label{thm:hausdorff_term_arbitrary}
  Let $P$ be a polygonal curve in $\R^2$ with $n$ vertices.

  \begin{itemize}[nosep]
  \item In $O(n\log n)$ time we can construct a data structure of size
    $O(n\log n)$ so that given a query segment $\overline{ab}$,
    $\dhd(P,\overline{ab})$ can be computed in $O(\log^2 n)$ time.
  \item In $O(n^2)$ time we can construct a data structure of size
    $O(n^2)$ so that given a query segment $\overline{ab}$,
    $\dhd(P,\overline{ab})$ can be computed in $O(\log n)$ time.
  \end{itemize}
\end{theorem}

\subsection{The backward pair distance term}
\label{sub:arb_orient_backward_pair}

Let $(p_i,p_j) \in \Skew$ be an ordered pair. 
We
restrict  $\xdistance_{pq}(\alpha,\beta)$ to the interval
of $\alpha$ values for which $(p_i,p_j)$ is a backward pair with
respect to orientation $\alpha$. Hence, each $\xdistance_{pq}$ is a
partially defined, constant algebraic degree, constant complexity,
bivariate function. The backward pair distance $\DB$ is the upper
envelope of $O(n^2)$ such functions. 
This envelope has
complexity $O(n^{4+\eps})$, for some arbitrarily small $\eps > 0$, and
can be computed in $O(n^{4+\eps})$
time~\cite{sharir1994almost}. Evaluating $\DB(\alpha,\beta)$ for some
given $\alpha, \beta$ takes $O(\log n)$
time. The following lemma, together with Theorem~\ref{thm:hausdorff_term_arbitrary} then gives an $O(n^{4+\eps})$ size data structure that supports $O(\log n)$ time \frechet distance queries.

\begin{lemma}
  \label{lem:backward_pair_arb}
  Let $P$ be an $n$-vertex polygonal curve in $\R^2$.
  In $O(n^{4+\eps})$ time we can construct a size $O(n^{4+\eps})$
  data structure so that given a query segment $\overline{ab}$,  $\DB(\overline{ab})$ can be computed in
  $O(\log n)$ time.
\end{lemma}

\section{Arbitrary Orientation Subcurve queries}
\label{sec:arb_Subcurve_Queries}

Next, we show how to support querying against subcurves $P[s,t]$ of $P$ in $O(\log^4 n)$ time as well. We use the same approach as for the horizontal query segment case: we store the vertices of $P$ into the leaves of a range tree where 
each internal node $\nu$ corresponds to some canonical subcurve
$P_\nu$,
so that any subcurve $P[s,t]$ can be represented by
$O(\log n)$ nodes.

\subparagraph{The Hausdorff distance term.} Since computing the
directed Hausdorff distance is decomposable, using this approach with
the data structure of Theorem~\ref{thm:hausdorff_term_arbitrary} immediately gives us a data structure that allows us to compute
$\dhd(P[s,t],\overline{ab})$ in $O(\log^2 n)$ time. Since the space
usage satisfies the recurrence $S(n) = 2S(n/2) + O(n^2)$, this uses
$O(n^2)$ space in total.

\subparagraph{The backward pair distance term.} By storing
the data structure of Lemma~\ref{lem:backward_pair_arb} at every node
of the tree, we can efficiently compute the contribution of the
backward pairs inside each of the $O(\log n)$ canonical subcurves that
make up $P[s,t]$. However, as in
Section~\ref{sec:Horizontal:_Subtrajectories}, we are still
missing the contribution of the backward pairs from different
canonical subcurves. We again store additional data structures at
every node of the tree that allow us to efficiently compute this
contribution at query time.

Let $S$ and $T$ be (the vertices of) two such canonical subcurves,
with all vertices of $S$ occurring before $T$ along $P$. Analogous to
Section~\ref{sec:Horizontal:_Subtrajectories} we will argue that
for some given $\alpha$ and $\beta$ the functions
$x \mapsto h_{\LLa{T}}(x,\alpha x + \beta)$ and
$x \mapsto h_{\RRa{S}}(x,\alpha x + \beta)$ 
are monotonically increasing and decreasing, respectively, and that the
intersection point of (the graphs of) these functions corresponds to
the contribution of the backward pairs in $S \times T$. So, our goal
is to build data structures storing $S$ and $T$ that given a query
$\alpha,\beta$ allow us to efficiently compute the intersection point
of these functions. As we will argue next, we can use the data
structure of Lemma~\ref{lem:querying_h_arbitrary_orient_large_space}
to support such queries in $O(\log^2 n)$ time.


We generalize some of our earlier geometric observations to
arbitrary orientations. Let $p,q$ be vertices of $P$, and let
$S$ and $T$ be subsets of vertices of $P$.
We define
\begin{align*}
\delta'_{pq}(\alpha,\beta) &=
\min_x \max \left\{ h_{\LLa{q}}( (x, \alpha x + \beta) ),
h_{\RRa{p}}( (x, \alpha x + \beta ) ) \right\}
& \text{ and } \\
\DB^{S \times T}(\alpha,\beta) &= \max \left\{ \xdistance'_{pq}(\alpha,\beta) \mid (p,q) \in (S\times T) \cap \Skew \right\}.
\end{align*}
and prove the following lemmas (by essentially rotating the plane so that the query segment becomes horizontal, and applying the appropriate lemmas from earlier sections): 

\begin{lemma}\label{lem:db_minx}
	Let $P$ be partitioned into two subcurves $S$ and $T$ with all
	vertices in $S$ occurring on $P$ before the vertices of $T$. We have
	that
	\[
	\DB(\alpha,\beta) = \DB^{P \times P}(\alpha,\beta) =
	\max \left\{
	\DB^{S \times S}(\alpha,\beta), \DB^{T \times T}(\alpha,\beta),
	\DB^{S \times T}(\alpha,\beta)
	\right\}.   \]
\end{lemma}

\begin{proof}
	Consider the rotation that transforms the line $y=\alpha x + \beta$ into a horizontal line, keeping $a$ to the left of $b$.
	Applying this rotation to $P$ and $\overline{ab}$ produces an instance of the horizontal query problem.
	Since the same rotation is applied to all vertices of $P$, the set of backward pairs and the relative distances from the query segment to them remain unchanged.
	Therefore, we can apply the results for the horizontal case, in
	particular Lemma~\ref{lem:rewrite}, to obtain that
	$\DB(\alpha,\beta) = \DB^{P \times P}(\alpha,\beta)$.
	
	For the second equality we use some basic set theory together with
	the facts that (i) $S$ and $T$ are a partition of $P$ and thus
	$P \times P = (S \times S) \cup (T \times T) \cup (S \times T) \cup
	(T \times S)$ and (ii) that all vertices of $S$ occur before $T$ and
	thus $(T \times S) \cap \Skew = \emptyset$.
\end{proof}

\begin{lemma}\label{lem:intersect_is_opt_arbitrary_orient}
	Let $S$ and $T$ be subsets of vertices of $P$, with $S$ occurring
	before $T$ along $P$ and let $\alpha,\beta$ denote some query
	parameters. The function $x \mapsto h_{\LLa{T}}(x,\alpha x + \beta)$
	is monotonically increasing, whereas
	$x \mapsto h_{\RRa{S}}(x,\alpha x + \beta)$ is monotonically
	decreasing. These functions intersect at a point
	$(x^*,\alpha x^* + \beta)$, for which
	$\DB^{S \times T}(\alpha,\beta) = h_{\LLa{T}}(x^*,\alpha x^* +
	\beta) = h_{\RRa{S}}(x^*,\alpha x^* + \beta)$.
\end{lemma}

\begin{proof}
	Consider the rotation that transforms the line $y=\alpha x + \beta$ into a horizontal line, keeping $a$ to the left of $b$.
	Note again that this preserves distances.
	It follows then from
	Observation~\ref{obs:monotonic} that
	$x \mapsto h_{\LLa{T}}(x,\alpha x + \beta)$ and
	$x \mapsto h_{\RRa{S}}(x,\alpha x + \beta)$ are monotonically
	increasing and decreasing, respectively. Furthermore, by
	Lemma~\ref{lem:intersection} they intersect in a single point
	$(x^*,y)=(x^*,\alpha x^* + \beta)$, at which
	$\DB^{S \times T}(\alpha,\beta) = h_{\RRa{S}}=(x^*,y) =
	h_{\RRa{S}}(x^*,y)$.
	%
	%
	%
\end{proof}

\subparagraph{Querying $\DB^{S \times T}(\alpha,\beta)$.} 
Consider the predicate
$Q(x) = h_{\LLa{T}}(x,\alpha x + \beta) < h_{\RRa{S}}(x,\alpha x +
\beta)$. It follows from
Lemma~\ref{lem:intersect_is_opt_arbitrary_orient} that there is a
single value $x^*$ so that $Q(x) = \textsc{False}$ for all $x < x^*$
and $Q(x) = \textsc{True}$ for all $x \geq x^*$. Moreover, $x^*$
realizes $\DB^{S \times T}(\alpha,\beta)$. By storing $S$ and $T$,
each in a separate copy of the data structure of
Lemma~\ref{lem:querying_h_arbitrary_orient_large_space}, we can
evaluate $Q(x)$, for any value $x$, in $O(\log n)$ time. We then
use parametric search~\cite{megiddo1983parametric} to find $x^*$ in
$O(\log^2 n)$ time.

\begin{lemma}\label{lem:query_intersect}
	Let $S, T$ be subsets of vertices of $P$ such that all vertices in $S$
	precede all vertices in $T$, stored in
	the data structure of
	Lemma~\ref{lem:querying_h_arbitrary_orient_large_space}. For any
	query $\alpha,\beta$ we can compute
	$\DB^{S \times T}(\alpha,\beta)$ in $O(\log^2 n)$ time.
\end{lemma}

\begin{proof}
	We treat $x^*$ as a variable, and evaluate $Q$ on the
	(unknown) value $x^*$. While doing so, we maintain an interval that is
	known to contain $x^*$. Initially this interval is $\R$ itself. When
	the algorithm to evaluate $Q(x^*)$ reaches a comparison involving
	$x^*$, we obtain a low degree polynomial in $x^*$ (as all comparisons
	in the query algorithm of
	Lemma~\ref{lem:querying_h_arbitrary_orient_large_space} test if the
	query point lies left or right of some line, or compare the Euclidean
	distance between two pairs of points). We compute the (constantly
	many) roots of this polynomial, and evaluate $Q$ again at each
	root. This shrinks the interval known to contain $x^*$, and allows the
	evaluation of $Q(x^*)$ to proceed. When the evaluation of $Q(x^*)$
	finishes, the interval known to contain $x^*$ has shrunk to a single
	point, $x^*$, or $x^*$ can be computed from it by solving one more
	equation in constant time. Evaluating $Q$ takes $O(\log n)$ time, and
	thus encounters at most $O(\log n)$ comparisons. For each such
	comparison we again evaluate $Q$ at a constant number of roots, taking
	$O(\log n)$ time each. Hence the total time required to compute $x^*$
	is $O(\log^2 n)$. Given $x^*$ we can obtain
	$\DB^{S \times T}(\alpha,\beta) = h_{\LLa{T}}(x^*,\alpha x + \beta)$
	in $O(\log n)$ time.
\end{proof}

Note that the parametric search approach we used here is an
$O(\log n)$ factor slower compared to the approach we used for
horizontal queries only (Lemma~\ref{lem:persistent_sweep_ds}).

For every node $\nu$ of the recursion tree on $P$ we store: (i)
the data structure of Lemma~\ref{lem:backward_pair_arb} built on its
canonical subcurves $P_\nu$, and (ii) the data structure of
Lemma~\ref{lem:querying_h_arbitrary_orient_large_space} built on the
vertices of $P_\nu$. The total space usage of the data structure
follows the recurrence $S(n)=2S(n/2) + O(n^{4+\eps})$, which solves to
$O(n^{4+\eps})$.
To query the data structure with some subcurve $P[s,t]$ from some
vertex $s$ to a vertex $t$ we again find the $O(\log n)$ nodes whose
canonical subcurves together define $P[s,t]$, query the
Lemma~\ref{lem:backward_pair_arb} data structure for each of them, and
run the algorithm from Lemma~\ref{lem:query_intersect} for each
pair. The running time is dominated by this last step, as this
requires $O(\log^2 n)$ time for each pair, and we have $O(\log^2 n)$
pairs to consider. Hence, the total running time is $O(\log^4 n)$. As
before, the procedure can be easily extended to the case where
$s$ and $t$ lie on the interior of an edge. We conclude:

\begin{lemma}\label{lem:backward_pair}
	Let $P$ be a polygonal curve in $\R^2$ with $n$ vertices. There is
	an $O(n^{4+\eps})$ size data structure that can be built
	in $O(n^{4+\eps})$ time such that given an arbitrary query segment
	$\overline{ab}$ and two query points $s$ and $t$ on $P$ it can
	report $\DB^{P[s,t]\times P[s,t]}(\alpha,\beta)$ in $O(\log^4 n)$ time. Here, $\eps > 0$
	is an arbitrarily small constant.
\end{lemma}

Since we can compute all four terms $\|s-a\|$, $\|t-b\|$,
$\dhd(P[s,t],\overline{ab})$, and
$\DB^{P[s,t]\times P[s,t]}(\alpha,\beta)$ in $
O(\log^4 n)$ time, it follows that we can efficiently answer \frechet distance queries
against subcurves.

\subsection{Space vs Query time tradeoff}
\label{sub:Space_vs_Query_time_tradeoff}

We can use our approach for subcurve queries from
Section~\ref{sec:arb_Subcurve_Queries} to obtain a space vs query time
trade off for queries against the entire curve. Let $k \in [1..n]$ be
a parameter. We trim the recursion tree on $P$ at a node $\nu$ of size
$O(k)$. Let \T denote the resulting tree (i.e. the top $\log (n/k)$
levels of the full recursion tree), and let $L(\T)$ denote the set of
leaves of \T, each of which thus corresponds to a subcurve of length
$O(k)$. Let $\ell(\nu)$ and $r(\nu)$ be the left and right child of $\nu$,
respectively. 
By repeated application of the second equality in Lemma~\ref{lem:db_minx} we have that

\[
  \DB^{P \times P}(\alpha,\beta)
  = \max \left\{ \max_{\nu \in \T} \DB^{P_{\ell(\nu)} \times P_{r(\nu)}}(\alpha,\beta)
               , \max_{\nu \in L(\T)} \DB^{P_\nu \times P_\nu}(\alpha,\beta)
             \right \}
\]

At every leaf of \T we now store the data structure of
Lemma~\ref{lem:backward_pair_arb}, and at every internal node the data
structure of
Lemma~\ref{lem:querying_h_arbitrary_orient_large_space}. The space
required by all Lemma~\ref{lem:backward_pair_arb} data structures is
$O((n/k)k^{4+\eps}) = O(nk^{3+\eps})$. The total size for all
Lemma~\ref{lem:querying_h_arbitrary_orient_large_space} data
structures follows the recurrence $S(n) = 2S(n/2) + O(n^2)$ which
solves to $O(n^2)$. Hence, the total space used is $O(nk^{3+\eps} +
n^2)$. The preprocessing time is $O(nk^{3+\eps} + n^2)$ as well.

To answer a query $(\alpha,\beta)$ we now query the
Lemma~\ref{lem:backward_pair_arb} data structures at the leaves of \T
in $O(\log k)$ time each. For every internal node $\nu$ we use
Lemma~\ref{lem:query_intersect} to compute the contribution of
$\DB^{P_{\ell(\nu)} \times P_{r(\nu)}}(\alpha,\beta)$ in $O(\log^2 n)$
time. Hence, the total query time is
$O((n/k)\log k + (n/k)\log^2 n) = O((n/k)\log^2 n)$.
So, e.g., choosing $k=n^{1/3}$ yields an  $O(n^{2+\eps})$ size data
structure supporting $O(n^{2/3}\log^2 n)$ time queries. We can extend
this idea to support subcurve queries in $O((n/k)\log^2 n + \log^ 4n)$
time as well, giving us the following result:

\begin{lemma}
  \label{lem:tradeoff_subcurve}
  Let $P$ be a polygonal curve in $\R^2$ with $n$ vertices, and let
  $k \in [1..n]$ be a parameter. In $O(nk^{3+\eps} + n^2)$ time we can
  construct a data structure of size $O(nk^{3+\eps} + n^2)$ so that
  given a query segment $\overline{ab}$, $\DB(\overline{ab})$ can be
  computed in $O((n/k)\log^2 n)$ time. If, in addition we are also given
  two points $s$ and $t$ on $P$,
  $\DB^{P[s,t] \times P[s,t]}(\overline{ab})$ can be computed in
  $O((n/k)\log^2 n + \log^4 n)$ time.
\end{lemma}
\begin{proof}
	We can also apply the time space
	trade off in case of subcurve queries. We use essentially the
	structure as above: for all nodes in the full recursion tree that
	represent subcurves of length at most $k$ we store both the
	Lemma~\ref{lem:backward_pair_arb} and the
	Lemma~\ref{lem:querying_h_arbitrary_orient_large_space} data
	structures, whereas for the topmost nodes in the tree (with a
	subcurve of size $>k$) we store only the
	Lemma~\ref{lem:querying_h_arbitrary_orient_large_space} data
	structure. The space usage remains $O(nk^{3+\eps}+n^2)$.
	
	We can again find a set $Q$ of $O(\log n)$ nodes $\nu$ whose subcurves
	together represent the query subcurve $P[ps,t]$.  We then have
	
	\begin{equation}
	\label{eq:subcurves}
	\DB^{P[s,t] \times P[s \times t]}(\alpha,\beta) =
	\max\left\{ \max_{\nu \in Q} \DB^{P_\nu \times P_\nu}(\alpha,\beta)
	, \max_{\mu,\nu \in Q, \mu \neq \nu} \DB^{P_\mu \times P_\nu}(\alpha,\beta)
	\right\}.
	\end{equation}
	
	\begin{figure}[tb]
		\centering
		\includegraphics{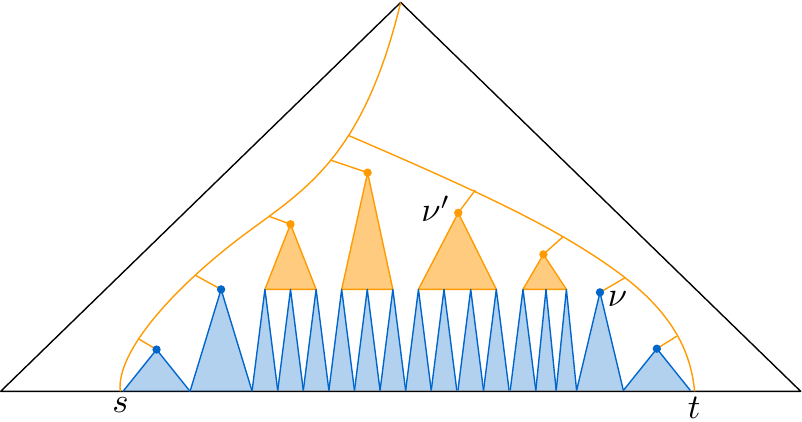}
		\caption{A query against a subcurve $P[s,t]$ selects $O(\log n)$
			nodes. When (the subcurve $P_\nu$ of) such a node $\nu$ is small
			enough (the blue nodes) we can directly query their
			Lemma~\ref{lem:backward_pair_arb} data structures. When the
			subtree is too large (e.g. $\nu'$) we visit its top part (in
			orange) computing the their contribution to
			$\DB^{P[s,t]\times P[s,t]}(\alpha,\beta)$ using
			the Lemma~\ref{lem:querying_h_arbitrary_orient_large_space} data structures until we
			reach the (blue) subtrees of size $O(k)$. The total number of
			nodes visited is $O(n/k)$.
		}
		\label{fig:tradeoff_tree}
	\end{figure}
	
	Observe that for each of distinct nodes $\mu,\nu \in Q$ we can compute
	$\DB^{P_\mu \times P_\nu}(\alpha,\beta)$ in $O(\log^2 n)$ time using
	the Lemma~\ref{lem:querying_h_arbitrary_orient_large_space} data
	structures stored at nodes $\mu$ and $\nu$
	(Lemma~\ref{lem:query_intersect}). Hence, computing the contribution
	of the second term in Equation~\ref{eq:subcurves} takes $O(\log^4 n)$
	time. To compute $\DB^{P_\nu \times P_\nu}(\alpha,\beta)$ for some
	$\nu \in Q$ there are two cases. If $P_\nu$ has length at most $k$ and
	thus we can compute the term directly by querying the
	Lemma~\ref{lem:backward_pair_arb} data structure of node $\nu$ in
	$O(\log k)$ time. If $P_\nu$ has size more than $k$, we essentially
	run the query algorithm from the beginning of this section on the
	subtree rooted at $\nu$ (See Fig.~\ref{fig:tradeoff_tree}): i.e. we
	traverse the tree starting from the root to find a minimal set of
	nodes $Q_\nu$ that store a Lemma~\ref{lem:backward_pair_arb} data
	structure, and whose associated subcurves make up $P_\nu$. We query
	their associated data structures, as well as the
	Lemma~\ref{lem:querying_h_arbitrary_orient_large_space} data
	structures of their subtree ancestors. Observe that the subcurve of
	each node of $Q_\nu$ has size $\Theta(k)$ (otherwise we would have
	picked its parent instead). Therefore, the total size of all sets
	$Q_\nu$ over all visited nodes $\nu$ is $O(n/k)$. The total number of
	nodes whose Lemma~\ref{lem:querying_h_arbitrary_orient_large_space}
	data structure we query is thus also $O(n/k)$, and hence the cost of
	querying these nodes is $O((n/k)\log^2
	n)$. 
\end{proof}

Since computing $\dhd(P[s,t],\overline{ab})$ can be done in
$O(\log^2 n)$ time using only $O(n^2)$ space, we thus established
Theorem~\ref{thm:arbitrary_full_curve_ds}. Once again, it is possible
to make the query time proportional to the complexity of $P[s,t]$
rather than to $n$.

\section{Applications}
\label{sec:Applications}

In this section we discuss how to apply our result to improve the results of some prior publications.

\subparagraph{Curve simplification.} Our tools can be used for local curve simplification under the \frechet distance.
In curve simplification, the input is a polygonal curve $P = (p_1, p_2, \ldots p_n)$ with $n$ vertices and some parameter $\delta$ and the goal is to compute a polygonal curve $S$, whose vertices are vertices of $P$, that is within distance $\delta$ of $P$ for some distance metric and has a minimum number of vertices.
The simplification is a \emph{local simplification} if for every edge $e_{ij}$ in $S$ from $p_i$ to $p_j$, the distance between $e_{ij}$ and $P[i, j]$ is at most $\delta$. 
Curve simplification is well-studied within computational geometry. Recent examples that use the \frechet distance include algorithms for global curve simplification by Kerkhof~\etal~\cite{kerkhof19global_curve_simpl}, a polynomial-time algorithm for computing an optimal local simplification for the \frechet distance by van Kreveld~\etal
~\cite{kreveld20hausd_frech}, and an  $O(n^3)$ lower bound on computing a local simplification in higher dimensions  using $L_p$ ($p \neq 2$) norms by Bringmann and Chaudhury~\cite{bringmann21polyl}.

The state-of-the-art  approach to compute a local $\delta$-simplification using \frechet distance and the $L_2$ norm in 2D is the Imai-Iri line simplification algorithm~\cite{imai1998computational} (see also Godau~\cite{godau91}).
This algorithm considers all $O(n^2)$ edges between vertices of $P$ and computes for every edge $e_{ij}$ from $p_i$ to $p_j$ ($i < j$), the \frechet distance between $e_{ij}$ and $P[p_i, p_j]$ in total $O(n^3)$ time. They assign the corresponding \frechet distance as a weight to $e_{ij}$. 
This results in  $O(n^2)$ weighted edges (links).
Computing a minimum
link path from $p_1$ to $p_n$ in this graph takes $O(n^2)$ time, and gives the desired simplification.
Thus, 
a local $\delta$-simplification can be computed in $O(n^3)$ time and $O(n^2)$ space.

We can improve this state-of-the-art algorithm by applying Theorem~\ref{thm:arbitrary_full_curve_ds}.
Specifically, if we choose our parameter $k =n^{1/2}$,
we can construct the corresponding data structure in $O(n^{2.5 + \eps})$ time and space (where $\eps$ is an arbitrarily small positive constant).
For each edge $e_{ij}$, we can compute the \frechet distance between $e_{ij}$ and the subcurve $P[p_i, p_j]$ in $O(\sqrt{n} \log^2 n)$ time, and we conclude:

\begin{theorem}
  \label{thm:local_simplification}
  Let $P$ be a polygonal curve in $\R^2$ with $n$ vertices, and let
   $\delta>0$. We can compute a local $\delta$-simplification
  of $P$ with respect to the \frechet distance using $O(n^{5/2+\eps})$
  time and space.
\end{theorem}

\subparagraph{\frechet distance queries under translation.}
There are many geometric pattern matching application where one would like to find a transformation (e.g., translation or rotation) that minimizes the Fr\'echet distance between two input curves. One typical example is handwriting recognition~\cite{SKB07}.
Our results can be used to compute a data structure for Fr\'echet distance queries under translation.
 Given are two polygonal curves $P$ and $Q$ in the plane, the goal is to find an optimal translation of $Q$ such that the Fr\'echet distance between the translated version of $P$ and $Q$ is minimized. This problem can be solved in $O((nm)^3(n+m)^2\log(n+m))$ time~\cite{alt01translationFrechet} where $n$ and $m$ are the number of vertices of $P$ and $Q$, respectively.
 Recently, Gudmundsson et al.~\cite{gudmundssonfrechet2020} (full version~\cite{gudmundsson2021translation}) studied the query version of this problem, where the goal is to preprocess $P$, such that given a query curve $Q$ and two points $s$ and $t$ on $P$, one can find the translation of $Q$ that minimizes the \frechet distance between $P[s, t]$ and $Q$ efficiently.
 They study this query version in a restricted setting, where $Q$ is a
 horizontal segment.  Their data structure uses $O(n^2\log^2 n)$ space
 and allows for $O(\log^{32} n)$ time queries.  
By applying our data structure, we obtain the following result:

 \begin{theorem}\label{thm:frechet_translated}
 Let $P$ be a polygonal curve in $\R^2$ with $n$ vertices. There is an $O(n \log^2 n)$ size data structure that can be built in $O(n \log^2 n)$ time such that given any two points $s, t \in P$ and a horizontal query segment $\overline{ab}$, one can report the translation of $\overline{ab}$ that minimizes its \frechet distance to $P[s, t]$ in  $O(\log^{12} n)$ time.
 \end{theorem}

 \begin{proof}
The approach by Gudmundsson et al.~\cite{gudmundssonfrechet2020} essentially uses four levels of parametric search~\cite{megiddo1983parametric} to turn a \frechet distance data structure into a data structure that can find the translation of $\overline{ab}$ that minimizes the distance to $P[s, t]$.

Specifically, suppose that you have access to a data structure that for a horizontal query segment $\overline{ab}$ can decide if the \frechet distance between
$P[s', t']$ and $\overline{ab}$ is at most $R$ in $Q(n)$ time (where $s'$ and $t'$ are vertices of $P$).

The authors then use the data structure for four levels of parametric search.
The first two levels are used to, given two points $s, t$ on $P$ that are not necessarily vertices, compute the \frechet distance between $\overline{ab}$ and the subcurve $P[s, t]$ (as opposed to the \frechet distance between $\overline{ab}$ and $P$).
The third level decides (given a fixed $x$-coordinate for the point $a$ of $\overline{ab}$) the vertical translation of $\overline{ab}$ that minimizes the \frechet distance to $P[s, t]$.
The fourth level decides the arbitrary translation of $\overline{ab}$ that minimizes the \frechet distance to $P[s, t]$
, by deciding, for a given $x$-coordinate, whether the endpoint $a$ of $\overline{ab}$ lies left or right of this $x$-coordinate.

Each parametric search squares the running time of the decision algorithm that it has access to, and their final solution runs in $O( Q(n)^{16} )$ time.
By replacing the
data structure of de Berg~\etal~\cite{de2017data} by
Theorem~\ref{thm:horizontal_subcurve_ds} we note that we can skip the first two levels of
parametric search (since our data structure already supports arbitrary
subcurves $P[s,t]$),
getting the total query time down to $O( Q(n)^{4} )$. 
With our data structure, we can decide if the \frechet distance between
$P[s, t]$ and $\overline{ab}$ is at most $R$ in $O(\log^3 n)$ time. Hence our total running time is $O( (\log^3 n) ^ 4) = O( \log^{12} n )$.
\end{proof}

Since the Theorem~\ref{thm:horizontal_subcurve_ds} data structure is
essentially used as a black-box, replacing it with
Theorem~\ref{thm:arbitrary_full_curve_ds} then yields a data structure
supporting arbitrarily oriented query segments.

 \begin{theorem}\label{thm:translation_general}
   Let $P$ be a polygonal curve in $\R^2$ with $n$ vertices, let
   $k \in [1..n]$, and let $\eps >0$ be an arbitrarily small
   constant. There is an $O(nk^{3+\eps}+n^2)$ size data structure that
   can be built in $O(nk^{3+\eps}+n^2)$ time such that given a
   query segment $\overline{ab}$ and two points $s$ and $t$
   on $P$, it can report the translation of $\overline{ab}$ that
   minimizes its \frechet distance to $P[s, t]$ in
   $O((n/k)^4\log^{8}n+\log^{16} n)$ time.
 \end{theorem}

 \subparagraph{\frechet distance queries under translation and
   scaling.}  We answer an open question by Gudmundsson et al. by
 showing how to find a horizontal segment that minimizes its \frechet
 distance to a given polygonal curve $P$. We start with the case
that the vertical position of the segment is given and fixed. That is,
given a height $y'$, we want to determine a horizontal segment
$\overline{ab}$ with $y_{a}=y'$ such that the Fr\'echet distance
between $P$ and $\overline{ab}$ is minimized. Observe that
$x_a\leq x_{p_1}$ and $x_b\geq x_{p_n}$.  Recall that the \frechet
distance $\fd(P, \overline{ab})$ is the maximum of the four terms
$\| p_1-a \|$, $\| p_n-b \|$, $\dhd(P, \overline{ab})$, and
$\DB(y_a)$.  Note that $\DB(y_a)$ is independent of the length of
$\overline{ab}$.

We divide the set of vertices of $P$ into three subsets $P^{<}$,
$P^{>}$, and $P'$ where $P^{<}$ contains all vertices that have an
$x$-coordinate smaller than $x_{p_1}$, $P^{>}$ contains all vertices
that have an $x$-coordinate greater than $x_{p_n}$, and $P'$ contains
all vertices that have an $x$-coordinate in  $[x_{p_1},x_{p_n}]$. Note
that for all $p_i\in P^{<}$, the pair $(p_1,p_i)$ is a backward pair and for
all $p_j\in P^{>}$, the pair $(p_j,p_n)$ is a backward pair. 

\begin{lemma}
  \label{lem:optimal_length}
  Let $\overline{ab}$ be a horizontal segment with height $y'$ that minimizes the \frechet distance to a curve $P$. Then,
  $\fd(P, \overline{ab}) =\max \left\{ \dhd(P, {\ell}), \DB(y')
  \right\}$ where $\ell$ is a horizontal line at height $y'$.
\end{lemma}
\begin{proof}
Let $z=\max \left\{  \dhd(P, {\ell}),  \DB(y') \right\}$. Let $a'$ be the leftmost intersection point of $\ell$ with a circle with radius $z$ and center $p_1$  and let $b'$ be the rightmost intersection point of $\ell$ with a circle with radius $z$ and center $p_n$. 
We show that $\fd(P, \overline{a'b'})=z$ which proves the lemma.
Recall that $\fd(P, \overline{a'b'})=\max \left\{ \| p_1-a' \|,\| p_n-b' \|,\dhd(P, \overline{a'b'}),\DB(y')\right\}$.
Clearly, $\| p_1-a' \|=\| p_n-b' \|=z$ and $\DB(y')\leq z$. We need to show that $\dhd(P, \overline{a'b'})\leq z$.
As $x_a'\leq x_{p_1}$ and $x_b'\geq x_{p_n}$, $\dhd(P', \overline{a'b'})=\dhd(P', \ell)$ (for points in $P'$, the Hausdorff distance is the vertical distance to $\ell$). 

We show that $\dhd(P^<, \overline{a'b'}) \leq z$ by contradiction. Suppose instead that $\dhd(P^<, \overline{a'b'}) > z$. 
Because $\dhd(P, {\ell})\leq z$, the vertical distance from $P^<$ to $\ell$ is at most $z$. Because all points of $P^<$ lie to the left of $p_1$ and hence $b'$, points $p_j\in P^<$ with distance greater than $z$ to $\overline{a'b'}$ must lie left of $a'$, and their distance to $\overline{a'b'}$ is $\|p_j-a'\|$.
Because $\dhd(P^<, \overline{a'b'}) > z$, there exists some point $p_j\in P^<$ that lies left of $a'$ such that $\|p_j-a'\|>z$.
Because $p_j\in P^<$, $(p_1,p_j)$ is a backward pair.
However, $\|p_1-a'\|=z$ and $x_{p_j}\leq x_{a'}\leq x_{p_1}$,
so the
backward pair distance of $(p_1,p_j)$ is $ \|p_1-q\|$ for some
$q=(x^*,y')$ with $x^*<x_{a'}$ (recall that the backward pair distance is the minimum possible distance between a point at height $y'$ and both $p_1$ and $p_j$, and that the vertical distance from $p_1$ and $p_j$ to $\ell$ is at most $z$). 
Thus, the backward pair distance is strictly greater than $z$, which contradicts $\DB(y')\leq z$.

We can prove that $\dhd(P^>, \overline{a'b'}) \leq z$ using a symmetric argument.
\end{proof}

Note that given the value $z=\max\{ \dhd(P, {\ell}), \DB(y')\}$, we
can find a segment $\overline{a'b'}$ on $\ell$ with Fr\'echet distance
$z$ to $P$ in constant time.

Next, we want to determine the height $y^*$ such that the \frechet
distance between the curve $P$ and a horizontal segment
$\overline{ab}$ with $y_a=y^*$ is minimized. We show that this height
can be computed in $O(n\log ^2n)$ time, moreover, we present a data
structure with $O(n\log ^2n)$ size (that can be built in the same
time) that reports a horizontal segment that minimizes the \frechet
distance to a subcurve of $P$ in polylogarithmic time.

First, we describe a decision algorithm that decides whether the
height of a current candidate segment is larger or smaller than the
height $y^*$ of an optimal segment.  Let the height of the current candidate
segment be $y'$ and let $\ell$ be a line with height $y'$. Recall that
by Lemma~\ref{lem:optimal_length} the \frechet distance between $P$
and this segment is either determined by $\dhd(P, {\ell})$ or
$\DB(y')$.  We have the following cases:
\begin{itemize}
    \item the \frechet distance is determined by $\dhd(P, {\ell})$:
      If the point that has the largest vertical distance to $\ell$
      lies below $\ell$, then the optimal height has to be
      smaller than $y'$. If this point lies above $\ell$, the optimal height has to be larger than $y'$. If this point lies
      on $\ell$, we stop and the  \frechet distance is 0. 
    \item the \frechet distance is determined by $\DB(y')$: If the
      midpoint of a segment connecting the two points      
      of the backward pair determining
      $\DB(y')$ lies below $\ell$, then the optimal height has
      to be smaller than $y'$. If this midpoint lies above $\ell$, the
      optimal height has to be larger than $y'$. If this
      midpoint lies on $\ell$, we found the optimal
      height.
\end{itemize}

It can be the case that more than one term determines the current \frechet distance. Then we decide for each of these terms whether the next candidate height has to be larger or smaller than the current one. If the decisions are the same for all  these terms, we move the height in the corresponding 
direction. Otherwise, we stop, as moving the height in any direction will increase the \frechet distance.
(Note that this decision algorithm is quite similar to the one of Gudmundsson et al.~\cite{gudmundssonFrechetJournal}.)

The \frechet distance between a polygonal curve and a horizontal segment at
optimal height as a function of $y'$ is convex, and has complexity
$O(n\log n)$. By Lemma~\ref{lem:optimal_length} it is the maximum of
the backward pair distance, and the Hausdorff distance from $P$ to the
line at height $y'$. De Berg et al.~\cite{de2017data} already showed
that the backward pair distance is convex, and by
Lemma~\ref{lem:backward_pair_ds} it has complexity $O(n\log n)$. The
Hausdorff distance is determined only by the top and bottommost point
in $P$, and is also easily seen to be convex and of constant
complexity. Therefore, the maximum of these two functions is
also convex and has $O(n\log n)$ breakpoints; at most a constant
number per piece of $\DB(y')$.

It then follows that given $P$ we can compute a horizontal segment
that minimizes the \frechet distance to $P$ in $O(n\log^2 n)$ time: we
use Lemma~\ref{lem:backward_pair_ds} to compute $\DB$, and construct
the function representing $y \mapsto \max\{\DB(y),\dhd(P,\ell)\}$ in
$O(n\log n)$ time, and use the above binary search procedure to
find a height $y^*$ where this function is minimized, and an optimal
segment at height $y^*$ that realizes this \frechet distance.
(Note that we can achieve the same running time without performing a binary search. As there are  
only $O(n\log n)$ break points, we could compute for each of them explicitly the corresponding height and \frechet distance, go through all of them and report the minimum one.)

We can also support queries where we find a horizontal segment
minimizing the \frechet distance to a query subcurve $P[s,t]$ of
$P$. One option is to simply use parametric search with the above
algorithm as decision algorithm. We show that we can do slightly
better by explicitly binary searching over the critical values.

\begin{theorem}\label{thm:frechet_stretched_translated}
    Let $P$ be a polygonal curve in $\R^2$ with $n$ vertices. There is
    an $O(n \log^2 n)$ size data structure that can be built in $O(n
    \log^2 n)$ time such that given any two points $s, t \in P$, one
    can report a horizontal segment 
    that minimizes
    its \frechet distance to $P[s, t]$ in  $O(\log^{4} n)$
    time.
\end{theorem}

\begin{proof}
  In this case we store the data structure from
  Lemma~\ref{lem:bisectorSubDS}, and a data structure that can report
  the minimum and maximum $y$-coordinate of a query curve $P[s,t]$.

  To answer a query we now wish to binary search (using the above
  decision procedure) on the breakpoints of the function
  $y \mapsto \max\{\DB^{P[s,t] \times
    P[s,t]}(y),\dhd(P[s,t],\ell)\}$. We will simply binary search on the
  breakpoints $y_1,\dots,y_m$ of $\DB^{P[s,t] \times P[s,t]}$; for each
  candidate breakpoint $y_i$ we find its successor $y_{i+1}$, and
  explicitly compute the $O(1)$ additional breakpoints in
  $[y_i,y_{i+1}]$ contributed by the Hausdorff term (by
  intersecting $\DB^{P[s,t] \times P[s,t]}$ with the function describing
  the Hausdorff distance).

  As in Lemma~\ref{lem:bisectorSubDS}, the function
  $\DB^{P[s,t] \times P[s,t]}$ is not stored explicitly, but represented
  by $O(\log^2 n)$ functions $\DB^{P_\nu \times P_\mu}$. Therefore, we
  cannot access the breakpoints $y_1,\dots,y_m$ directly; some of them are
  not even represented explicitly. Instead, we use a two phase approach
  in which we maintain an interval $I^{\nu,\mu}$ for each function
  $\DB^{P_\nu \times P_\mu}$ that is known to contain the value $y^*$
  that we are searching for. We keep shrinking these intervals until
  each function has no more breakpoints in its interval. In the second
  phase we can then explicitly compute $\DB^{P[s,t] \times P[s,t]}$
  inside $\bigcap I^{\mu,\nu}$ by constructing the upper envelope of the
  functions $\DB^{P_\nu \times P_\mu}$ (restricted to their intervals). 
  Since the total complexity of all functions (restricted to their intervals) 
  is only $O(\log^2 n)$, this takes $O(\log^2 n\log\log n)$ time. 
  We can now find $y^*$ (which is guaranteed to lie in $\bigcap I^{\mu,\nu}$) 
  by explicitly traversing (the breakpoints of) $\DB^{P[s,t] \times P[s,t]}$.

  In the first phase, we will shrink the intervals by simultaneously
  binary searching over all $O(\log^2 n)$ functions. Let $n^{\nu,\mu}$
  be the complexity of $\DB^{P_\nu \times P_\mu}$ restricted to the
  interval $I^{\nu,\mu}$, and let $N = \sum n^{\nu,\mu}$ be the total
  remaining complexity. Assume that for every function
  $\DB^{P_\nu \times P_\mu}$ we can find a \emph{halving point}
  $y^{\mu,\nu} \in I^{\mu,\nu}$ for which the complexity of
  $\DB^{P_\nu \times P_\mu}$ in $I^{\mu,\nu}$ before and after
  $y^{\mu,\nu}$ is roughly halved. That is, let $c \in (1,2]$ be a
  constant such that the complexity of $\DB^{P_\nu \times P_\mu}$ in
  both intervals is at most $n^{\mu,\nu}/c \leq N/c$.

  We can then binary search as follows: for each function (that still
  has breakpoints inside its interval), compute its halving point
  $y^{\mu,\nu}$, and its complexity $n^{\mu,\nu}$. Consider the
  halving points in increasing order, and for $y^{\mu,\nu}$ compute
  the sum $M^{\mu,\nu}$ of the
  complexities before $y^{\mu,\nu}$, i.e.,
  \[
    M^{\mu,\nu} = \sum_{y^{\mu',\nu'} \leq y^{\mu,\nu}} n^{\mu',\nu'}.
  \]

  Now there are two cases, either there is a halving value
  $y^{\mu,\nu}$ for which $M^{\mu,\nu}$ lies in the range
  $[N(\frac{1}{2}-\frac{1}{2c}), N(\frac{1}{2}+\frac{1}{2c})]$, or
  there is a single function $\DB^{P_\nu \times P_\mu}$ that has at
  least complexity $N(1/c)$. In either case, we pick $y^{\mu,\nu}$ to
  reduce the size of the intervals.  That is, we evaluate
  $\DB^{P[s,t] \times P[s,t]}$ at $y^{\mu,\nu}$ and use the decision
  procedure to test if $y^*$ occurs before or after $y^{\mu,\nu}$
  (this means we may have to construct the piece of
  $\DB^{P[s,t] \times P[s,t]}$ at $y^{\mu,\nu}$ and compute its
  intersection with the function representing the Hausdorff distance
  term). We now argue that in both cases we can discard at least a
  constant fraction of the total complexity.

  In the former case, we pick a halving point for which $M^{\mu,\nu}$
  lies in the range
  $[N(\frac{1}{2}-\frac{1}{2c}), N(\frac{1}{2}+\frac{1}{2c})]$. If we
  discard the first half of the intervals, i.e., up to $y^{\mu,\nu}$,
  we can discard the first halves for all functions for which
  $y^{\mu',\nu'} < y^{\mu,\nu}$. Since these functions have total
  complexity $M^{\mu,\nu} \geq N(\frac{1}{2}-\frac{1}{2c})$ and we
  discard (at least) a $(1-\frac{1}{c})$ fraction of each function, we
  reduce the complexity by at least a fraction
  $(\frac{1}{2}-\frac{1}{2c})(1-\frac{1}{c})=(\frac{1}{2}-\frac{1}{c}+\frac{1}{2c^2})$. Note
  that for any $c \in (1,2]$ this value is strictly positive, i.e., for
  $c=2$ we would discard a fraction of $1/8$, for something like
  $c=3/2$ we would discard a $1/18$ fraction. Similarly, if we discard the
  second half of the intervals, we can discard the second halves for
  all functions for which $y^{\mu',\nu'} \geq y^{\mu,\nu}$. These
  functions have total complexity at least
  $N-N(\frac{1}{2}+\frac{1}{2c}) =
  N(\frac{1}{2}-\frac{1}{2c})$. Since again we  discard a
  $(1-\frac{1}{c})$ fraction from each function, we reduce the
  complexity by at least a $(\frac{1}{2}-\frac{1}{c}+\frac{1}{2c^2})$
  fraction.

  In the later case, we discard at least a $(1-\frac{1}{c})$ fraction
  of $\DB^{P_\nu \times P_\mu}$. Since this function has complexity at
  least $N/c$, we reduce the total complexity by at least
  $N\frac{1}{c}(1-\frac{1}{c}) = N(\frac{1}{c}-\frac{1}{c^2})$.

  In either case we thus reduce the complexity by a constant
  fraction. Hence, after a total of $O(\log N)$ rounds, there are no
  more interior breakpoints, and we can start phase two. Apart from
  finding the halving points and complexities, each such round takes
  $O(\log^3 n)$ time, since we have to evaluate
  $\DB^{P[s,t] \times P[s,t]}$ at $y^{\mu,\nu}$ using
  Lemma~\ref{lem:bisectorSubDS}.

  All that remains is to argue that we can also find the halving
  points and complexities (in at most $O(\log n)$ time per round) for
  each function. The functions $\DB^{P_\nu \times P_\nu}$ are directly
  represented using binary search trees, and hence finding a halving
  point is trivial. The functions $\DB^{P_\nu \times P_\mu}$, with
  $\nu \neq \mu$, are represented using
  Lemma~\ref{lem:persistent_sweep_ds} data structures. These
  structures are essentially also just binary search trees, hence the
  same applies here.

  Since the total complexity $N$ of all $\DB^{P_\nu \times P_\mu}$ functions
  is at most $O(n\log^2 n)$, our binary search finishes in $O(\log n)$
  rounds. The theorem now follows.
\end{proof}

By using parametric search instead of the binary search and replacing
our data structure for horizontal query segments with the one for
arbitrarily oriented query segment, we immediately get the following
result.

 \begin{theorem}\label{thm:translation_stretched_general}
    Let $P$ be a polygonal curve in $\R^2$ with $n$ vertices, let
    $k \in [1..n]$, and let $\eps >0$ be an arbitrarily small
    constant. There is an $O(nk^{3+\eps}+n^2)$ size data structure that
    can be built in $O(nk^{3+\eps}+n^2)$ time such that given a
    query slope $\alpha$ and two points $s$ and $t$
    on $P$, one can report a segment  with slope $\alpha$  that
    minimizes its \frechet distance to $P[s, t]$ in
    $O((n/k)^2\log^{4}n+\log^{8} n)$ time.
\end{theorem}

\section{Concluding Remarks}
\label{sec:Concluding_Remarks}

We presented data structures for efficiently computing the \frechet
distance of (part of) a curve to a query segment.  Our results improve
over previous work for horizontal segments and are the first for
arbitrarily oriented segments. However, we are left with the challenge
of reducing the space used for arbitrary orientations. There are two
main issues. The first issue is that even for a small interval of
query orientations (e.g., one of the $O(n^2)$ angular intervals defined
by lines through a pair of points) it is difficult to limit the number
of relevant backward pairs to $o(n^2)$. The second issue is how to
combine the backward pair distance values contributed by various
subcurves. For (low algebraic degree) univariate functions, the upper
envelope has near linear complexity, whereas for bivariate
functions the complexity is near quadratic. The combination of these issues
makes it hard to improve over the somewhat straightforward
$O(n^{4+\eps})$ space bound we build upon.



\paragraph*{Acknowledgements.}{\small This work started at Dagstuhl
  workshop 19352, ``Computation in Low-Dimensional Geometry and
  Topology.''  We thank Dagstuhl, the organizers, and the other
  participants for a stimulating workshop.
  T. Ophelders was supported by the Dutch Research Council (NWO) under project no.\ VI.Veni.212.260.
  Research of Schlipf was supported by the Ministry of Science, Research and the Arts
  Baden-Württemberg (Germany).
  R. Silveira was partially supported by MCINN through project PID2019-104129GB-I00/MCIN/AEI/10.13039/501100011033.
}

\bibliographystyle{abbrv}
\bibliography{frechet_dss}

\end{document}